\documentclass[12pt,article,fleqn]{scrartcl}
  
\RequirePackage{amsthm,amsmath,natbib}
\RequirePackage{amssymb,amstext,dsfont, mathtools}
\usepackage[pdfpagemode=UseOutlines ,plainpages=false
,hypertexnames=false ,hyperindex=true,colorlinks=true]{hyperref}
	\makeatletter%
	\Hy@breaklinkstrue%
	\makeatother%
\usepackage{color}
\definecolor{darkred}{rgb}{0.6,0.0,0.1}
\definecolor{darkgreen}{rgb}{0,0.5,0}
\definecolor{darkblue}{rgb}{0,0,0.5}
\hypersetup{colorlinks ,linkcolor=darkblue ,filecolor=darkgreen
,urlcolor=darkblue ,citecolor=black}

\usepackage{ bm}
\usepackage{bbm}

\renewcommand{\cite}{\citet}
\usepackage{graphicx}
\usepackage{pgf, tikz}
\usepackage{enumerate}
\usepackage{subcaption}
\usepackage{animate} 
\usepackage{multirow}
\usetikzlibrary{positioning}

\def\Ex{\mathop{{\mathbf E}\null}\nolimits}%
\newcommand{\PP}{\mathbb{P}}

\def\1{\mathop{\mathbbm 1}\nolimits}
\DeclarePairedDelimiter{\mynorm}{\lVert}{\rVert}

\definecolor{dgreen}{rgb}{0,0.5,0}
\definecolor{dblue}{rgb}{0,0,0.9}
\definecolor{dred}{rgb}{0.6,0.0,0.1}
\definecolor{dgold}{rgb}{0.5,0.3,0.0}
\definecolor{dvio}{rgb}{0.6,0.3,0.5}
\definecolor{gray}{rgb}{0.5,0.5,0.5}

\newcommand{\Perp}{\perp \!\!\! \perp}

\usepackage{setspace}
\theoremstyle{mysc}\newtheorem{prop}{Proposition}[section]
\theoremstyle{mysc}\newtheorem{coro}[prop]{Corollary}
\theoremstyle{mysc}\newtheorem{theo}[prop]{Theorem}
\theoremstyle{mysc}
\theoremstyle{mysc}\newtheorem{lem}[prop]{Lemma}
\theoremstyle{myex}\newtheorem{rem}{Remark}[section]
\theoremstyle{myex}\newtheorem{example}{Example}[section]
\theoremstyle{myex}

 \theoremstyle{mysc}\newtheorem{assA}{Assumption}
 \renewcommand{\theassA}{\arabic{assA}}

\newtheorem*{assumption*}{\assumptionnumber}
\providecommand{\assumptionnumber}{}
\makeatletter

\numberwithin{equation}{section}

  \author{{\textsc{Christoph Breunig}}\thanks{Department of Economics, Emory University, Rich Memorial Building, Atlanta, GA 30322, USA. Email:
\url{	christoph.breunig@emory.edu}}\\
{\small \textit{Emory University}}
  \and{\textsc{Stephan Martin}}\thanks{Deutsche Bundesbank; Frankfurt, Germany. The paper represents the author's personal opinion and does not necessarily reflect the views of the Deutsche Bundesbank or its staff.
  Humboldt-Universit\"at zu Berlin, Spandauer Stra\ss e 1, 10178 Berlin, Germany, e-mail: \url{stephan.martin@bdpems.de}}
  \\
  {\small \textit{Deutsche Bundesbank and HU Berlin}}
}
 \title{Nonclassical Measurement Error in the Outcome Variable\thanks{
We are thankful to seminar participants at CMStatistics/CFE in London, Humboldt-Universit\"at zu Berlin, Retreat of CRC TRR 190,  and UEA in Norwich for their helpful suggestions.
Financial support by Deutsche Forschungsgemeinschaft through CRC TRR 190 is gratefully acknowledged.}}

\begin{document}
   \maketitle

\begin{abstract}
 We study a semi-/nonparametric regression model with a general form of nonclassical measurement error in the outcome variable. We show equivalence of this model to a generalized regression model. Our main identifying assumptions are a special regressor type restriction and monotonicity in the nonlinear relationship between the observed and unobserved true outcome. Nonparametric identification is then obtained under a normalization of the unknown link function, which is a natural extension of the classical measurement error case. 
We propose a novel sieve rank estimator for the regression function and establish its rate of convergence.
 In Monte Carlo simulations, we find that our estimator corrects for biases induced by nonclassical measurement error and provides numerically stable results. 
We apply our method to analyze belief formation of stock market expectations with survey data from the German Socio-Economic Panel (SOEP) and find evidence for nonclassical measurement error in subjective belief data. 
\end{abstract}
\vskip .5cm
\begin{center}
nonparametric identification, special regressors, generalized regression, sieve estimation.
\end{center}

\newpage

\section{Introduction}

In empirical research, measurement error is a recurring issue. In recent years, much attention has been given to various forms of measurement error in the covariates of econometric models, whereas measurement error of the dependent variable is mostly ignored.
In many economic environments, measurement error of the dependent variable may be driven (in a nonlinear fashion) by the underlying variable. This nonclassical measurement error implies biased estimation results if not accounted for. 

This paper is concerned with semi-/nonparametric regression models where the dependent variable of interest $Y^*$ is generally not observed and only a possibly error-contaminated measurement $Y$ is observable. Specifically, $Y^*$ satisfies
\begin{align}\label{reg:model}
Y^*=g(X)+U, 
\end{align}
where the unknown function $g$ is of interest given observed covariates $X$ and unobservables $U$.
We study the \textit{nonclassical} measurement error case where $\Ex[Y|Y^*, X] \neq Y^*$. Hence, the regression function $g$ does in general not coincide with conditional expectations of observable variables and we cannot impose 
$g(x)= \Ex[Y|X=x]$. 

Nonparametric identification of our model relies on the availability of covariates which do not affect the measurement error directly. We impose such type of exclusion restriction on a subset $Z$ of the vector $X=(Z,W)$, where $W$ are additional controls. 
Under a monotonicity condition on the measurement error mechanism, we show in this paper that
model \eqref{reg:model} can be reformulated as a generalized regression model of the form
\begin{align*}
\Ex[Y|X=x]=H(g(x), w),
\end{align*}
where $H(\cdot,w)$ is a nonlinear, monotonic function for $w$ in the support of $W$.
Identification of the function $g$,  up to strictly monotonic transformations, immediately follows, which allows us to infer on economically relevant quantities such as the direction and shape of partial effects.

Under scale and location normalization of the unknown link function $H$, nonparametric identification of the regression function $g$ is obtained. We highlight that normalization of the link function $H$ is equivalent to imposing mild shape restrictions on the measurement error mechanism.
Additionally, our normalization conditions on the link function do not only naturally extend the classical measurement case but are also  satisfied if there is a range of $Y^*$ where measurement error is classical. Our nonparametric identification results build thus on intuitive assumptions without relying on high-level assumptions such as completeness, see \cite{hu2008}. 

We consider a sieve, rank-based minimum distance estimator and establish its asymptotic properties.
We derive the rate of convergence in $L^2$ sense of our estimator.
 We find that the sieve rank estimator generally suffers from ill-posedness in the convergence rate as the rank-based criterion function is not continuous in the usual $L^2$-norm. We develop the theory for the case where $W$ is discrete and provide an  extension  to allow continuous controls $W$ using kernel weights in the appendix of this paper.
 
We analyze the performance of the estimator in a Monte Carlo simulation study and in an empirical application using survey data. We apply our estimator to study belief formation with subjective belief data from the German Socio-Economic Panel innovation sample (SOEP-IS). Subjective belief data is known to be plagued by substantial measurement error and it is in general hard to justify that the measurement error is classical and thus not sensitive to the underlying true individual belief.
 We study the impact of an exogenous display of historic stock market returns provided to survey respondents prior to eliciting their belief on future returns. 
Applying our method, we find a monotonic and concave relationship between the historic information and stated beliefs indicating that individuals acknowledge the given information conservatively.  

\paragraph{Literature}
Our work ties into the literature on measurement error in observable variables of econometric models. 
The literature on measurement error in covariates is extensive, whereas measurement error in the outcome variable has received much less attention. For a review of models with errors in covariates, see e.g. \cite{Chen_ME_Rev} and \cite{Schennach_rev}. 
\cite{CHT_ME2005} develop a general way of accounting for measurement error in any variable of a class of semiparametric models once auxiliary data, e.g. from validation samples is available. However, this is hardly the case in most practical applications. 
Models focusing on nonclassical measurement error in the outcome side are rare. Chapter 3 of \cite{AbrevHausman99} considers a semiparametric model with a more simplistic measurement error mechanism. 
\cite{HoderleinWinter2010} and \cite{hoderlein2015} develop structural models of response error in surveys due to imperfect recall and derive testable implications for econometric analyses. The latter paper focuses on the role of rounding in individual reporting behavior which is also a more specific form of nonclassical measurement error.  

\cite{NadaiLewbel} allows for classical measurement error in the outcome variable that is correlated with an error in covariates. \cite{abrevaya2004} consider classical measurement error of the dependent variable in a transformation model.
Given we have a precise idea on the form of measurement error, a sizeable literature is usually available providing different strategies for identification. For instance a special case of nonclassical measurement error is selective non-response in the outcome variable, see e.g. \cite{2010Hault} or \cite{breunigEndSel18} and references therein.
A non-nested form of nonclassical measurement error are Berkson-type errors, see \cite{Berkson} and \cite[Section 6.3]{Schennach_rev}.

Our identifying assumptions lead us to the literature on generalized regression models as introduced in \cite{han1987} or the class of nonlinear index models in \cite{matzkin2007}. See also the model studied in \cite{JachoLewbel}.  
Estimation of such models often proceeds by rank-based estimation strategies, see \cite{han1987}, \cite{CavSherman}, \cite{khan2001}, \cite{shin2010} and \cite{abrevaya2011} which all consider parametric regression models with the exception of \cite{matzkin1991} who studies a nonparametric model with additional shape restrictions on the link function. A recent contribution studying rank estimators in a high-dimensional setting is \cite{FanRank20}. 
To the best of our knowledge, we are the first to study nonparametric M-estimation with rank-based criterion functions and to point out and illustrate the ill-posedness of the estimation problem.
\cite{JureckovaBernoulli16} study a different class of rank estimators in the context of a parametric model with measurement error in both regressors and outcome. Their the outcome error may not be nonclassical as in our general notion but can at most depend on observable regressors.

The remainder of the paper is organized as follows.
In Section \ref{Sec::Ident} we present our model setup and give a nonparametric identification result for features of the mean regression function when there is a form of nonclassical measurement error in the outcome variable. In Section \ref{Sec:Est} we introduce a sieve estimator with a rank based criterion function and establish its convergence. In Section \ref{Sec::MC} we analyze finite sample properties of the estimator in a Monte Carlo simulation study. Section \ref{Sec::Appl1} contains an application of our method to belief formation of stock market expectations. Appendix \ref{Sec::ConW} provides an extension to weighted sieve rank estimation, when control variables are continuous. All proofs are postponed to the Appendix \ref{Appendix:proofs}.

\section{Model Setup and Identification}\label{Sec::Ident}
We consider a nonparametric econometric model with measurement error in the outcome variable. The model we study is
\begin{equation}
Y^*=g(X)+U, \quad \label{model} 
\end{equation}	
where $Y^*$ is the scalar, outcome variable, $X$ is a $d_x$-dimensional vector of exogenous covariates, $U$ is a scalar error term, and $g$ a nonparametric function of interest. The outcome variable $Y^*$ is not observed by the researcher; only an error contaminated measurement $Y$ is available. 
We are primarily interested in the case where the error satisfies $\Ex[U|X]=0$ and thus $g$ is the unknown conditional expectation function of $Y^*$ given $X$.

Throughout the paper, we assume that the regressors $X$ can be decomposed such that $X=(Z', W')'$, where $Z$ has no direct effect on the measurement error and $W$ are control variables. Also we introduce the notation $g_w(\cdot)\equiv g(\cdot, w)$ for the regression function evaluated at a fixed $w$ in the support of $W$. We now provide conditions, which allow for nonparametric identification of $g_w$ up a strictly monotonic transformation.


\begin{assA}[Exclusion Restriction]\label{A:IV}
The observed outcome $Y$ is conditionally mean independent of $Z$ given $Y^*$ and $W$, i.e., $\Ex[Y|Y^*,Z, W]=\Ex[Y|Y^*, W]$.
\end{assA}
Assumption \ref{A:IV} rules out that $Z$ has a direct effect on the measurement $Y$ in conditional expectations.  Assumption \ref{A:IV} is generally weaker than assuming that the conditional distribution of $Y$ given $(Y^*, Z,W)$ does not depend on $Z$, which restricts $Z$ to have no information on $Y$ that is not captured by $(Y^*,W)$. 
Analogues exclusion restrictions are commonly imposed in the literature on nonclassical measurement error in covariates. In  \cite[Assumption 2 (ii)]{hu2008}, the distribution of the error-contaminated regressor is independent of instruments conditional on the latent regressor (see also \cite[Section 4.3]{Schennach_rev}). Assumption \ref{A:IV} is less restrictive than other exclusion restrictions found in the measurement error literature, see \cite[Assumption 2.1 (iii)]{benmosh2017}.
 
Conditions similar to Assumption \ref{A:IV} can also be found in the literature on selective non-response, which is a special case of nonclassical measurement error in the outcome. Individuals either report the outcome truthfully (response indicator $D=1$) or not at all ($D=0$) so the observed outcome in this case is $Y=D Y^*$. See also Remark \ref{rem:sel:response} below. An identifying assumption in \cite{2010Hault} and \cite{breunigEndSel18} is that $D \Perp X \;|\; (Y^*, W)$, which is related to Assumption \ref{A:IV}. 

In the following, we make use of the notation $h(Y^*, W)=\Ex[Y|Y^*, W]$. Assumption \ref{A:IV} implies the measurement error model 
\begin{align*}
Y=h(Y^*, W)+V,
\end{align*}
where $\Ex[V|Y^*,W]=0$. Consequently, Assumption \ref{A:IV}   implies conditional mean independence of the measurement error $V$ given the regression error $U$, that is, $\Ex[V | U]=0$. 
\begin{assA}[Monotonicity]\label{A:Mon}
For any $w\in\textsl{supp}(W)$, the function $h(\cdot, w)$ is weakly monotonic and non-constant over the support of $Y^*$.
\end{assA}

Assumption \ref{A:Mon} imposes that the expected observed outcome $Y$ is monotonic in the latent outcome $Y^*$ given $W$. This is trivially satisfied when the measurement error is classical, i.e., when $h$ does not depend on $W$ and is the identity. 
A similar monotonicity condition has also been imposed in the measurement error model in 
\cite[Example 3]{AbrevHausman99}.\footnote{In our notation \cite{AbrevHausman99} consider the error mechanism $Y=h(Y^*,V)$, with $\partial_y h(Y^*,V)> 0$, $\partial_v h(Y^*,V)> 0$  and $V \Perp (X, U)$. As we allow for heteroscedasticity in the measurement error model, condition $\partial_y h(Y^*,V)> 0$ may lead to one sided error restrictions.}
Note that $h$ does not need to be strictly monotonic which allows to consider models with rounding error in the outcome, see \cite{hoderlein2015}.
We discuss the plausiblity of Assumption \ref{A:Mon} in the context of the application in Section \ref{Sec::Appl1} in a setting with survey data.

\begin{assA}[Conditional Exogeneity]\label{A:Unobs} 
The conditional independence restriction $Z \Perp U\; |\; W$ holds.
\end{assA}
Assumption \ref{A:Unobs} imposes a conditional independence restriction of $Z$ and the regression error $U$. This condition is also known as conditional exogeneity assumption following \cite{white2010}. 
Independence assumptions can be restrictive,  but are often required in the measurement error literature (see, e.g. \cite{Hausmanetal91}, \cite{Schennach07}, \cite[Assumption 2.2]{benmosh2017}), or when accounting for endogeneity using control functions (see, e.g. \cite{Neweyetal99}).
We relax such restrictions by imposing independence to hold only conditional on control variables $W$.
Similar conditions are often employed for identification in the econometrics literature, see e.g. \cite{CKK} for nonparametric identification in a transformation model. Assumption \ref{A:Unobs} also corresponds to the unconfoundedness assumption in the treatment effects literature and is also closely related to the special regressor assumption, see \cite{Lewbel_SpecReg} for a review.
Next, we need the following set of regularity conditions. We introduce the notation $\textsl{supp}(V)$ for the support of a random vector $V$. 
\begin{assA}\label{regIdent} For any $w\in \textsl{supp}(W)$: 
	(i) the function $g_w$ is continuous;
	(ii)  and any $z_1, z_2 \in \textsl{supp}(Z)$ such that $g_w(z_1) < g_w(z_2)$ there exists $u \in \textsl{supp}(U)$ satisfying $h(g_w(z_1)+u, w) < h(g_w(z_2)+u, w)$; 
	(iii)  there is at least one variable $Z_{(1)}$ such that $Z=(Z_{(1)},Z_{(-1)})$ with $f_{Z_{(1)}| Z_{(-1)},W}(z_1|z_{-1},w) > 0 $ for all $(z_1, z_{-1})\in\textsl{supp}(Z)$.  
\end{assA}
Assumption \ref{regIdent} (ii) is a mild support condition on $U$ conditional on $W=w$. The unobservable $U$ must vary sufficiently to shift $g_w(Z)$ out of a flat region of $h$. The assumption is not required if $h$ is strictly monotonic in its first argument. Assumption \ref{regIdent}  (iii) requires $Z$ to contain at least one continuously distributed variable with sufficient variation.
If $Z$ is scalar then Assumption \ref{regIdent} (iii) may be replaced by  $f_{Z|W}(z|w) > 0 $ for all $z\in\textsl{supp}(Z)$. This rules out the case of $Z$ being a discrete scalar variable. 

Under the stated assumptions, now provide establish equivalence to the regression model \eqref{model} to a generalized regression model specified by the link function $H(g_w(z), w)=\Ex[h(g(z, W)+U, W) \;|\; W=w] $. Below, $\mathds{1}\{\cdot\}$ denotes the indicator function.

\begin{theo}\label{thm:ident}
	Let Assumptions \ref{A:IV}--\ref{regIdent} be satisfied, then for any $w \in \textsl{supp}(W)$
	it holds
	\begin{align}\label{eq:transform}
	\Ex[Y|X=x] = H(g_w(z), w),
	\end{align}
	where $H(\cdot,w)$ is strictly monotonically increasing  and  $g_w(z)$ maximizes the function
	\begin{align}\label{eq:transform:max}
	\mathcal Q(\phi,w)=\Ex[Y_1\mathds{1}\{ \phi(X_1) > \phi(X_2)\}   \;|\; W_1=W_2=w].
	\end{align}
	In particular,  the function $g_w(\cdot)$ is identified up to strictly increasing transformations.
\end{theo}
The model \eqref{eq:transform} falls into the class of generalized regression models studied by \cite{han1987}, \cite{matzkin1991}, and  \cite{CavSherman}.
Further note that nonclassical measurement error implies heterogeneous biases for the marginal effects. When  $\partial_z H(g_w(z), w)<1$ we obtain an \textit{attenuation bias} for the marginal effect $\partial_z g_w(z)$ and when $\partial_z H(g_w(z), w)>1$ we get an \textit{augmentation bias} for $\partial_z g_w(z)$.

Theorem \ref{thm:ident} implies identification of features of $g_w$ that are preserved under monotonic transformations. This includes the sign of partial effects, the ratio of two partial effects\footnote{Note that for $g(z_1, z_2)$ it holds that $\frac{\partial g}{\partial z_1}/\frac{\partial g}{\partial z_2} = \frac{\partial H(g)}{\partial z_1}/\frac{\partial H(g)}{\partial z_2}$ whenever these quantities and ratios are well-defined.} and properties such as quasi-concavity (-convexity) of the function. For the remainder of the paper we consider identification and estimation of $g_w$ in the point identified case.

We impose the following restriction on the model and the measurement error mechanism described by the function $H$.
	\begin{assA}\label{shapecons}
(i) The function $g_w$ is additively separable such that there exists a decomposition $Z=(Z_1, Z_{-1})$ such that  $g_w(Z)=m_w(Z_1) + l_w(Z_{-1})$ for some functions $m_w, l_w$. (ii) There exists $\{z_1, z_2\} \subset \textsl{supp}(Z)$ with $g_w(z_1)\neq g_w(z_2)$ and $\Ex[Y|Z=z, W=w]=\Ex[Y^*|Z=z, W=w ]$ for $z\in\{z_1, z_2\}$.
	\end{assA}
Assumption \ref{shapecons} (i) imposes an additive separable structure on the regression function $g_w$. Following the identification statement in Theorem \ref{thm:ident}, mere location and scale normalizations are not sufficient to point identify $g_w$. However, for any additive separable model this is the case, see also \cite{JachoLewbel}. Assumption \ref{shapecons} (ii) restricts the measurement error for at least to realizations of $Z$.
	 Assumption \ref{shapecons} (ii) is also in line with normalization requirements for identification under nonclassical measurement error. For instance, Assumption 5 of \cite{hu2008} requires some functional of the distribution of the measurement error conditional on the value of the true variable to be equal to the true variable itself, such as some quantile of $Y|Y^*=y^*$ to correspond to $y^*$. 
	 
Economic restrictions on the model can also be employed to sufficiently restrict the function space. 
We refer to the discussion in Sections 3.4 and 4.4 in \cite{matzkin2007} where several possible function spaces are discussed that can replace Assumption \ref{shapecons}(i). This includes the spaces of functions that are homogeneous of degree one or so called ``least-concave''
functions, see also \cite{matzkin1994}.
\cite{matzkin2007} shows that imposing homogeneity of degree 1 and a location normalization is sufficient for Assumption \ref{shapecons}.
Homogeneous functions are frequently encountered in microeconomics. Thus, in applications where the function $g$ has the structural interpretation of a production or cost function, homogeneity can be a reasonable restriction on the parameter space.

\begin{coro}\label{coro:ident}
	Let Assumptions \ref{A:IV}-- \ref{shapecons} (i) be satisfied, then the function $g_w$ is identified up to a location and scale normalization. If \ref{shapecons} (ii) is additionally satisfied then the function $g_w$ is point identified. 
\end{coro}
Corollary \ref{coro:ident} establishes identification of the regression function under normalization imposed in Assumption \ref{shapecons}.
The shape restrictions imposed in Assumption \ref{shapecons} imply a normalization of the unknown, nonparametric link function $H$, in contrast to nonparametric generalized regression models, where normalization is typically imposed on the unknown function of interest.

We neither restrict the support of the observed outcome $Y$, nor require continuity in the function $h(\cdot, w)$. Thus, we can also cover cases where the observed outcome is categorical or has mass points. This likely occurs in survey data as respondents tend to provide rounded values.
The following examples consider a generalization and special case of model (\ref{model}).

\begin{example}[Control function approach]
	We can also motivate the presence of $W$ in Assumption \ref{A:Unobs} as a control function.
	To this end we deviate for a moment from our previous notation and introduce the following triangular model
	\begin{align*}
	Y^*= & g(X)+U \\
	X =& m(Z, \eta) 
	\end{align*}
	where for simplicity $X$ is a one-dimensional endogenous covariate that may correlate with the model error $U$. The function $m$ is strictly monotonic in $\eta$ and $Z$ is an instrumental variable satisfying $	Z \Perp (U, \eta) $. 
	Under additional regularity conditions, following \cite[Theorem 1]{ImbensNewey09} it holds that
	\begin{align*}
	 X & \Perp U \;|\; W \quad \text{with} \\
	  W& = F_{X|Z}(X,Z)=F_\eta(\eta),
	\end{align*}
	where $F_V$ denotes the cummulative distribution function of a random variable $V$.
	As in Assumption \ref{A:IV} we impose $\Ex[Y | Y^*, Z, W]= \Ex[Y | Y^*, W]$. Thus, following Theorem \ref{thm:ident}, we obtain identification of the structural function $g$ up to a strictly monotonic transformation.
%
\end{example}

\begin{example}[Selective Nonresponse]\label{rem:sel:response}
	Consider a nonresponse model
	\begin{align*}
	Y &=D Y^* \\
	D &= \phi(Y^*,W, V),
	\end{align*}
	for some unknown function $\phi$, where the response indicator $D\in\{0,1\}$ is always observed and $Y^*$ is only observed if $D=1$. 
	This framework, where the response mechanism is mainly driven by the latent outcome $Y^*$ has been studied by \cite{2010Hault} and \cite{breunigEndSel18}.
	As long as the conditional mean function $h(Y^*, W)=P(D=1|Y^*, W)Y^*$ is monotonic in its first argument, the model is in accordance to Assumption \ref{A:Mon}. This holds e.g. when the conditional response probability function is monotonic and the support of $Y^*$ is bounded below\footnote{If $Y^*$ is bounded below, then $Y^*$ can be redefined such that without loss of generality $Y^* \geq 0$ and monotonicity of $h(Y^*, W)=P(D=1|Y^*, W)Y^*$ follows from taking the derivative.}.
	In this case, a completeness condition for nonparametric identification of the conditional selection probability $P(D=1|Y^*,W)$ (see \cite{2010Hault} and \cite{breunigEndSel18})  via conditional moment restrictions is not required. 
\end{example}

\section{Estimation and Asymptotic Properties}\label{Sec:Est}
In this section, we introduce a nonparametric sieve M-estimator with a simple, rank-based criterion function. For simplicity, we consider only the case where $W$ consists of discrete variables and defer the estimation with continuous $W$ to Appendix \ref{Sec::ConW}.
\subsection{The Sieve Rank Estimator}\label{Sec::Estim}


Our identification result builds on shape restrictions imposed on the measurement error mechanism, which imply identified moment conditions. Specifically, for a given $w$ we have from the identification statement in Theorem \ref{thm:ident} that the true $g_w$ maximizes the function
\begin{align*}
\mathcal Q(\phi,w)=\Ex[Y_1\mathds{1}\{ \phi(X_1) > \phi(X_2)\}   \;|\; W_1=W_2=w]. 
\end{align*}
Based on this population criterion, we now consider a sieve rank estimator, which implicitly accounts for imposed shape restrictions required for identification. 

We propose the following sieve rank estimator
\begin{align}\label{Estimator}
\widehat g_w &= \arg \max_{\phi \in \mathcal{G}_{K}} \mathcal Q_n(\phi, w) \;\; \text{where}\\
\mathcal Q_n(\phi,w) &:= \frac{2}{n(n-1)} \sum_{1 \leq i < j \leq n} Y_i   \mathds{1}\{W_i=W_j=w\} \mathds{1}\{\phi(Z_{i}) > \phi(Z_{j})\},  \nonumber
\end{align}
for some $K=K(n)$ dimensional sieve space $\mathcal{G}_{K}$. Here, the dimension parameter $K$ grows slowly with sample size $n$. 
For the special case where $W$ is absent, the criterion reduces to 
\begin{equation}\label{CS-criterion}
\mathcal Q_n(\phi)
= \frac{1}{n(n-1)}\sum_{i=1}^{n} Y_i \; \text{Rank}(\phi(Z_i)) ,
\end{equation}
where the rank function is defined as $\text{Rank}(\phi(Z_i))=\sum_{j \neq i}^n \mathds{1}\{\phi(Z_{i}) > \phi(Z_{j})\}$. This is a nonparametric version of the criterion of \cite{CavSherman}.

The specific choice of $\mathcal{G}_K$ hinges on the chosen normalization. Under a normalization of the link function $H$, see Corollary \ref{coro:ident},  we may consider a linear sieve space $\mathcal{G}_K=\{\phi: \phi(z)=\gamma_w'p^K(z)\}$.
Let $p^K=(p_{1}, \dots, p_{K})$ be a $K$- dimensional vector of known basis functions such as polynomials, splines or similar.
We can in principal also apply the general sieve estimation technique of \cite{Chen07} based on the conditional moment restriction $\Ex[Y|X=x]=H(g_w(z), w)$. This would require to estimate $H$ along with $g_w$ and nesting of two sieve spaces.
Our estimation strategy constructively arises from the identification argument and provides a simple direct estimate of $g_w$.  
We also directly leverage the monotonicity condition on $H$ in the estimation so there is no need to introduce additional shape-constraints.

\subsection{Convergence Rate}\label{Sec::Asympt}
In this section, we derive a rate of  convergence of the sieve rank estimator $\widehat g_w$ given  in \eqref{Estimator}.
To keep notation simple, we omit the controls $W$ entirely from the following analysis.  In this case, estimation amounts to maximizing the criterion in (\ref{CS-criterion}) from the previous section over a suitable sieve space. 

%
%

For the remainder of the paper we consider the centered criterion function 
\begin{align}\label{def:crit:centered}
\mathcal Q(\phi) &= \Ex\big[Y_i\big(\mathds{1}\{\phi(Z_i)>\phi(Z_j)\}- \mathds{1}\{g(Z_i)>g(Z_j)\}\big)\big]
\end{align}
where $g$ is the regression function satisfying the model equation (\ref{model}). Centering does not change the maximizer in the optimization problem and is thus without loss of generality.

Our analysis builds on a linearization of the nonlinear criterion function $\mathcal Q(\cdot)$. The first directional derivative of $\mathcal Q$ is equal to zero for any arbitrary direction and hence, we consider the second directional derivative which can be viewed as a quadratic approximation to the criterion function $\mathcal Q(\cdot)$.
Specifically, we introduce
\begin{align*}
 Q(\phi - g):= \frac{\partial^2}{\partial \tau^2} \mathcal Q(g+ \tau (\phi -g)) \;\Big{|}_{\tau=0}  
\end{align*}
denote the second directional derivative of the non-linear functional $\mathcal Q$ in the direction $\phi-g$. We  assume that the functional $Q(\cdot)$ is bi-linear and continuous. 
%
Below, we denote $L^2(Z)=\{\phi:\, \|\phi\|_{L^2(Z)}<\infty\}$ where $\|\phi\|_{L^2(Z)}:=\sqrt{\Ex\phi^2(Z)}$.

To account for the potential instability of the estimation problem, we introduce the sieve measure of ill-posedness
\begin{align*}
\tau_K = \sup_{\phi \in \mathcal G_K} \frac{\mynorm{\phi- \Pi_Kg}_{L^2(Z)} }{  Q(\phi- \Pi_Kg)}
\end{align*}
to account for the fact that the criterion function and the $L^2$-norm are generally not (locally) equivalent. If $\tau_K \to \infty$ as $K\to \infty$ the problem of estimating $g$ is  ill-posed in rate and additional regularization slows down convergence in the strong $L^2$- norm. In contrast to \cite{ChenPouzo12}, we rely on the second directional derivative in the denominator.

For the following assumption we introduce a local neighborhood of $g$ and define the space $\mathcal{G}_K^\delta = \{\phi \in \mathcal{G}_K: \mynorm*{\phi-g}_{L^2(Z)}< \delta\}$ with $\delta > 0$. 
\begin{assA}\label{A:Rate}
	(i) A random sample $\{(Y_i, Z_i)\}_{i=1}^n$ of $ (Y,Z)$ is observed;
	(ii) there exists $\Pi_K g\in\mathcal{G}_K$ such that $\|\Pi_K g-g\|_{L^2(Z)}=O(K^{-\alpha/d_z})$;	
	(iii) $\Ex[U^2] < \infty $ and $g \in L^2(Z)$; 
	(iv) for any $\phi$ in $\mathcal{G}_K^\delta$ there exists a constant $0<\eta< 1$ such that $|\mathcal Q(\phi)-Q(\phi-g)| \leq \eta \cdot Q(\phi -g)$;
	(v) the cdf of $g(Z)$ is Lipschitz continuous, i.e., $|F_{g(Z)}(a) - F_{g(Z)}(b)|\leq C|a-b|$ for some constant $C$ and any $a,b$;
	and (vi) $\tau_K \sqrt{K/n}=o(1)$.
\end{assA}
Assumption \ref{A:Rate} (ii) imposes regularity on the regression function $g$ via a sieve approximation error, see also \cite{Chen07} for examples.
Assumption \ref{A:Rate} (iv) is also known as the tangential cone condition and implies that $\mathcal Q(\phi)$ is locally equivalent to $Q(\phi -g)$ which is a typical condition required to derive the convergence rate for sieve estimators; see \cite[Assumption 4.1(ii)]{ChenPouzo12} and also \cite{Dunker2011}. 
Assumption \ref{A:Rate} (v) amounts to a local continuity assumption for the kernel of an empirical process, see e.g. \cite[Condition 3.8]{Chen07}. Assumption \ref{A:Rate} (vi) restricts the growth of $K$ relative to the sieve measure of ill-posedness $\tau_K$ and is required for consistency, see Lemma \ref{App::Lemma2}.   

\begin{rem}[Illustration of Ill-Posedness]\label{rem:ill:posed}
To give an insight on the source of ill-posedness, note that
\begin{align*}
\mathcal Q(\phi)=\Ex\left[Y_i\left(F_{g(Z_i)|Y_i}(g(Z_j)) - F_{\phi(Z_i)|Y_i}(\phi(Z_j))   \right) \right]
\end{align*}  
which shows that if there is little variation in the distribution of $F_{g(Z)|Y}$ for variations of g then the ill-posed inverse problem becomes more severe. This is further illustrated by the following lemma where we study a special case for which we can derive $Q$ analytically and give sufficient conditions for Assumption \ref{A:Rate} (iv).
\begin{lem}\label{Lem::AddSep}
		Consider the additive separable model $g(Z)=Z_{1} + \widetilde g(Z_{2})$ with bivariate $Z=(Z_{1}, Z_{2})$. Then Assumption \ref{A:Rate} (iv) is satisfied if $f^{'}_{Z_{1} | Z_{2}}$ is uniformly bounded away from zero and $f^{''}_{Z_{1} | Z_{2}}$ is uniformly bounded above.
\end{lem}
The special case outlined in Lemma \ref{Lem::AddSep} illustrates the behavior of $\tau_K$. If the density $f_{Z_{21} | Z_{22}}$, that is the conditional density of the separable covariate, is flat in the relevant support, we may encounter the case that the criterion $\mathcal{Q}$ is close to zero for candidate functions that are arbitrarily far away from the true function in the $L^2$- sense.
\end{rem}

 We further illustrate this issue in a Monte Carlo simulation study in Section \ref{Sec::MC}, where we show that the estimation problem becomes more difficult  as $f_{Z_{21} | Z_{22}}$ becomes more flat. We are now in a position to provide a general rate of convergence of our sieve rank estimator $\widehat g$.

\begin{theo}\label{thm::rate}
	Let Assumptions \ref{A:IV}-\ref{A:Rate} be satisfied. It holds that
	\begin{align*}
     \mynorm*{\widehat g-g}_{L^2(Z)}=O_p\Big( \max \Big\{ \tau_K \sqrt{\frac{K}{n}}, \; K^{-\alpha/d_z} \Big\}\Big)
     \end{align*}
\end{theo}
The proof of Theorem \eqref{thm::rate} makes use of a representation of second-order U-processes as empirical processes following \cite{ClemLugosi08}.
To the best of our knowledge, this is the first convergence rate result for nonparametric M-estimators with a rank-based criterion function in the presence of ill-posedness.

The next corollary provides concrete rates of testing
when the dimension parameter $K$ is chosen to level variance and square bias under classical smoothness conditions.
We call our model \textit{mildly ill-posed} if: $\tau_k \sim k^{\gamma/d_z}$ with $\gamma > 0$ and \textit{severely ill-posed} if: $\tau_k\sim\exp(k^{\gamma/d})$, with $\gamma > 0$.\footnote{If $\{a_n\}$ and $\{b_n\}$ are sequences of positive numbers, we use the notation $a_n \lesssim b_n$ if $\limsup_{n\to\infty}a_n/b_n<\infty$ and $a_n\sim b_n$ if $a_n\lesssim b_n$ and $b_n\lesssim a_n$.}
\begin{coro}
	Let Assumptions \ref{A:IV}-\ref{A:Rate} be satisfied. 
	\begin{itemize}
		\item[1.] Mildly ill-posed case: setting $K \sim n^{d_z/d_z+2\gamma+2\alpha}$ yields
		\begin{align*}
	\mynorm*{\widehat g-g}_{L^2(Z)}=O_p(n^{-\alpha/(2\alpha+2\gamma+d_z)}).
		\end{align*}
		\item[2.] Severely ill-posed case: setting $K\sim \log(n)^{d/\gamma}$ yields
		\begin{align*}
		\mynorm*{\widehat g-g}_{L^2(Z)}=O_p(\log(n)^{-\alpha/\gamma}).
		\end{align*} 
	\end{itemize}   
\end{coro} Both convergence rates are the optimal rates for ill-posed problems. As outlined in the discussion following Lemma \ref{Lem::AddSep}, the severity of the ill-posedness will generally depend on the chosen normalization and features of the data. 

\section{Monte Carlo Simulation Study}\label{Sec::MC}

This section demonstrates how nonclassical measurement errors in the outcome alters mean regression results in finite samples and shows the usefulness of our approach to correct for such biases. We compare regression function estimates obtained from simply ignoring the measurement error with our estimator, which accounts for the presence of the error. Throughout this section, simulation results are based on a sample of size of $n=1000$ and  1000 Monte Carlo iterations. 

We consider the following data generating process
\begin{align*}
Y^* & =Z_1+g(Z_2)+U \\
Y  &= h(Y^*)+V,
\end{align*}
where $Z_1 \sim \mathcal N(1,\sigma^2)$, $Z_2 \sim \mathcal U[-3, 3]$ independent of each other, $g(\cdot)=\sin(\cdot)$ and the error terms $(U,V)\sim\mathcal N(0,I_2)$. Here, $I_2$ is the 2-dimensional identity matrix and for the standard deviation of $Z_1$ we choose $\sigma=1$, which will be varied later. In the above model, $g$ is identified up to a location normalization. Analogously we could specify a linear or nonlinear function on $Z_1$ and impose an additional scale normalization on $g$. 
The function $h$ in the measurement error equation is chosen as
\begin{align*}
h(Y^*)= \begin{cases}
         q_{0.7}+b(Y^*-q_{0.7}) \;\;&\text{if}\;\; Y^* > q_{0.7}\\
         Y^*, \;\;&\text{if}\;\; q_{0.3} \le Y^* \le q_{0.7}  \\
         q_{0.3}-a(q_{0.3}-Y^*) \;\;&\text{otherwise}\
         \end{cases}
\end{align*}
where $q_{0.3}, q_{0.7}$ denote the $30\%$- and $70\%$-quantile of $Y^*$ (determined via numerical approximation). The setup is analogous to a  typical survey data setting with over- or underreporting in the tails of $Y^*$, whereas the center of the distribution is not affected. The scalars $a,b$ are chosen to vary the magnitude of measurement error. 
\begin{figure}[!htbp]
	\centering
	\includegraphics[width =13cm]{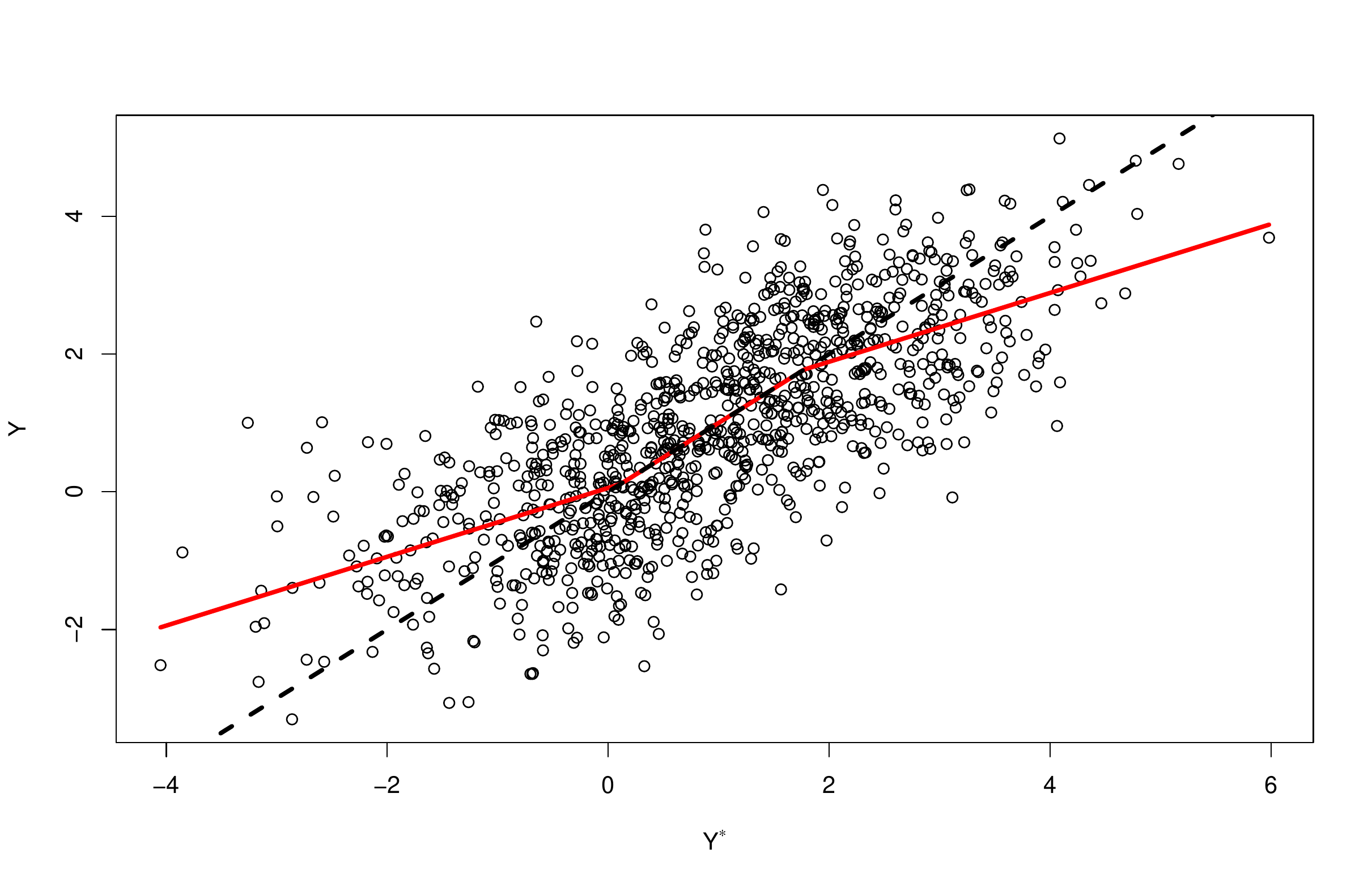}
		\caption{
		 Realizations of $Y^*,Y$ when $a=b=0.5$ based on a random draw of size $n=1000$. 
		The red solid line depict the function $h$ and the black dashed line the $45^\circ$ line.}\label{Plot_h}
\end{figure}

Figure \ref{Plot_h} illustrates the effects of the measurement error for the case $a=b=0.5$. 
We show the realizations of $Y$ and $Y^*$ for a specific draw of the data generating process and plots the function $h$. We compare the measurement error function $h$ (depicted as red solid line) with the setup of  classical measurement error, which is captured by the $45^\circ$ line (depicted as black dashed line). 

We implement the sieve rank estimator $\widehat g$ given in \eqref{Estimator} using a linear sieve space with B-spline basis functions of order 3 with 2 interior knots that are placed according to quantiles of the empirical distribution. Thus we have $K=4$. The elements of the sieve space are normalized at the point $(0,0)$ which is the correct value of the true function $\sin(\cdot)$ at $0$. This normalization can also be rationalized as utilizing prior knowledge on the measurement error mechanism in the sense of Assumption \ref{shapecons} (ii). For instance, we can expect that ignoring the measurement error results in estimates that are close to the true function $g$ in the center of the distribution of $Z_2$. Figure \ref{Plot_SimRes} shows the sieve rank estimates $\widehat g$ and compares them to a nonparametric series regression that does not account for nonclassical measurement error in the outcome using the same order and the same knot placement as for $\widehat g$. For the latter estimator the same choice of basis functions and tuning parameters is adopted.  
\begin{figure}[h!]
    \centering
	\includegraphics[width =14cm]{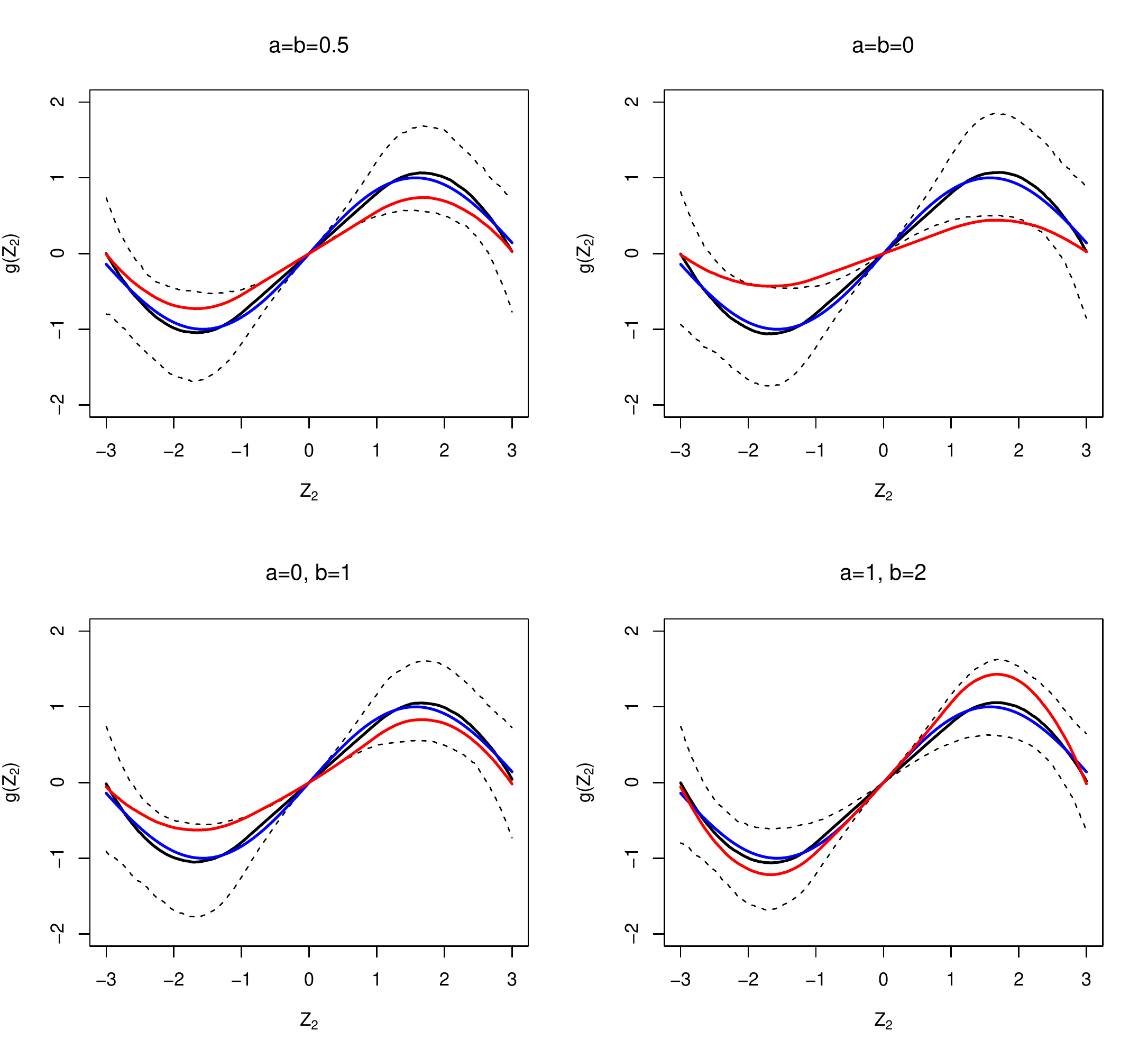}
		\caption{ Estimation results normalized to go through the coordinate $(0,0)$:
		 Solid black line is the median of our sieve rank estimator $\widehat g$, solid red line is the median of  a series estimator with same B-splines specification, solid blue line shows true $g(\cdot)$ function, and  dashed black lines are the 0.95 and 0.05 quantiles over all Monte Carlo rounds.}\label{Plot_SimRes}
\end{figure}
%
%
We study different values for $a,b$ amongst which is the severe case $a=b=0$ which essentially implies that at some point the measurements $Y$ are merely random fluctuations around a constant value\footnote{Additionally we perform Kolmogorov-Smirnov tests to test the null hypothesis that $Y$ and $Y^*$ follow the same probability distribution on every drawn sample of the MC study. In the $a=b=0.5$ setting we reject the null on a $5\%$ - level only once in 1000 samples and in the $a=b=0$ case we reject the null in 966 cases. Thus in the strong ME setting, $Y$ and $Y^*$ have different marginal distributions in contrast to the mild ME setting, where differences are virtually undetectable.}. 
We observe from the results in Figure \ref{Plot_SimRes} that our estimation strategy results in an accurate estimate of $g$ in any of the cases, whereas ignoring the measurement error yields estimates with a sizeable bias in the tails of $Z_2$. In the severe setting depicted in the right panel, ignoring measurement error results in a rather flat estimate which is significantly different from the sieve rank estimator. 

The data generating process chosen here is in line with the model in Lemma \ref{Lem::AddSep} and thus allows us to study the degree of ill-posedness in the convergence rate of the estimator. As pointed out in the discussion following Lemma \ref{Lem::AddSep}, the behavior of the sieve measure of ill-posedness $\tau_K$ is governed by the conditional density $f_{Z_1|Z_2}$.  If the density $f_{Z_1|Z_2}$ is flat over the relevant support, $\tau_K$ diverges faster and the ill-posedness is more severe. 

Table \ref{table:mse} below shows mean squared errors of function estimates across different standard deviations of the separable covariate $Z_1$ which affects the slope of the density $f_{Z_1|Z_2}$. For small standard deviations, the conditional density $f_{Z_1|Z_2}$, i.e., here $f_{Z_1}$ by full independence, will be rather flat over most of the support.
For small standard deviations of $Z_1$, the MSE increases more severely with $K$ as compared to large standard deviations. This illustrates that the degree of ill-posedness of the estimation problem is more severe whenever the slope of the density $f_{Z_1|Z_2}$ is small.
\begin{table}[ht]
	\renewcommand{\arraystretch}{1.15}
	\centering
	\begin{tabular}{|c|c||cccc|}
		\hline
		St. Dev. of $Z_1$ &	$Z_2 \sim \mathcal{U}[-c,c]$ & \multicolumn{4}{c|}{$\text{MSE}(\widehat g)$ for sieve dim.}\\
		$\sigma$& c & $K=3$ & $K=4$ & $K=5$ & $K=6$   \\\hline\hline
		\multirow{2}{*}{0.5} & 1 & 0.02209 & 0.06843 & 0.17294 & 0.52289	\\
		& 3  & 0.02389 & 0.05982 & 0.17068 & 0.62054 \\\hline
		\multirow{2}{*}{1} & 1 & 0.01579& 0.04293& 0.09087& 0.20775	\\
		& 3  & 0.01807 & 0.04783 & 0.08118 & 0.19650 \\\hline       
		\multirow{2}{*}{2} & 1 &0.01489 & 0.04316 & 0.09514 & 0.20622	\\
		& 3  & 0.01640 & 0.04580 & 0.08593 & 0.19877 \\\hline
	\end{tabular}
	\caption{ Results for the $\text{MSE}(\widehat g)$ for varying values of the standard deviation $\sigma$ of $Z_1$ and the range $c$ of $Z_2$. 
	}\label{table:mse}		
\end{table}
Additionally we see that this is not the case when the distribution of $Z_1$ is fixed and the dispersion of $Z_2$ is varied. This confirms that the ill-posedness in this setting is not driven by the distribution of $Z_2$ in this setting.
%
%

\section{Application: Beliefs on Stock Market Returns}\label{Sec::Appl1}
Subjective beliefs on stock market returns are an important determinant in economic models that seek to explain stock market participation and portfolio choice, see e.g. \cite{HuckWeiz19} and the references therein. 
Subjective belief data, however,  is known to be prone to a large degree of measurement error, see the discussion and references in \cite{GaudeckerME}. 



We study the impact of historic return information on subjective beliefs of future stock market returns. We account for nonclassical measurement error in the outcome variable by applying our sieve rank method and contrast the results to a model where we simply ignore measurement error in the outcome. 

We use novel data from the innovation sample of the 2017 wave of the German Socio Economic Panel (SOEP-IS), which contains survey questions on individual beliefs on future stock market returns.
In the interviews, respondents are asked their expectations on the DAX, Germany’s prime blue chip stock market index,   in one, two, ten and thirty years with respect to the current level. They are asked to provide a direction of the change (increase or decrease) as well as a percentage change.

Prior to elicitation of their beliefs, individuals obtain information about historical DAX returns. Two observations of the time series of yearly DAX returns from $1951$ to $2016$ are randomly drawn and presented to the respondent. Afterwards they are asked to report their beliefs on how the DAX changes in the next year (in percentage points).
\begin{table}[h]
	\renewcommand{\arraystretch}{1.15}
	\centering
	\begin{tabular}{|c||cccccc|}
	\hline
		      & Min. & 1. Quant & Median & Mean & 3. Quant. & Max.   \\\hline\hline
		$Y$&  -50.00  & 1.00 & 4.00 &  3.55 &  7.00 & 130.00   \\
		$Z_1$& -43.94  & -6.08 &  11.36 &  14.77 &  29.06 & 116.06 \\
		$Z_2$ &  -43.94 & -6.08 & 13.99 & 17.13 & 34.97 & 116.06  \\\hline
	\end{tabular} 
	\caption{Summary Statistics (all units are percentage points)}\label{table:sum}
\end{table}

In this application, we are interested in the effect of the historical DAX information on the individuals expected DAX return in one year.
Let $Y^*$ denote the individual true belief on the DAX return in one year and let $Z_1,Z_2$ be the two treatment variables, i.e., the randomly drawn historical returns. The reported belief is denoted by $Y$. 
We consider the following flexible additively separable model
\begin{align}\label{App_Model1}
Y^* = g_1(Z_1)+g_2(Z_2) + g_3(Z_1 \cdot Z_2) + U, \;\;\text{where}\;\; Z \Perp U .
\end{align}
It is difficult to rationalize a classical measurement error assumption a priori. 
Various forms of nonclassical measurement error may occur in this setting: (i) Respondents may tend to provide rounded values instead of precise beliefs, (ii) respondents may systematically over- or underreport their beliefs, e.g., individuals with extreme beliefs may resort to reporting more modest values, or (iii) the reporting may additionally depend on variables $W$ such as certain cognitive skills or personality traits like patience or perseverance. Note that by the experimental design $Z_1,Z_2$ and $W$ are credibly fully independent so there is no need to specify the variables in $W$ or to apply our weighted sieve rank estimator. 
\begin{figure}[!h]
	\centering
	\includegraphics[width=12cm]{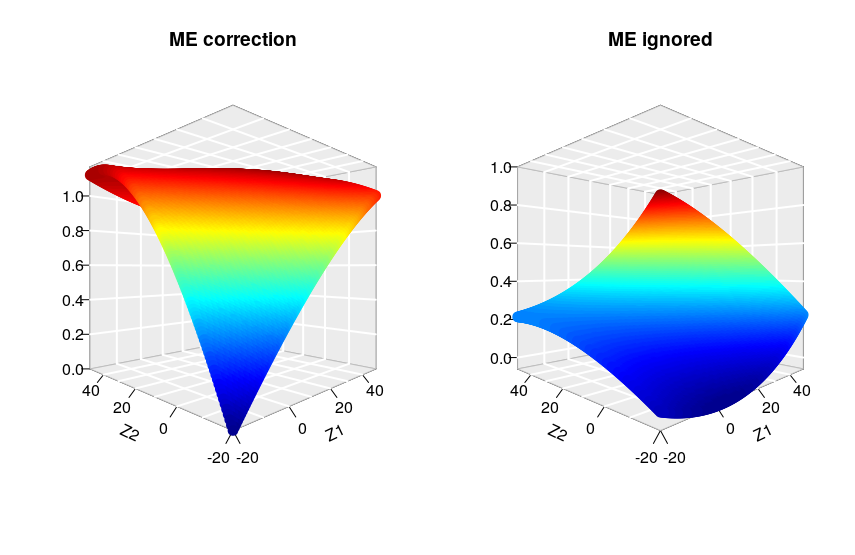}
	\vskip -1.2cm
	\includegraphics[width=12cm]{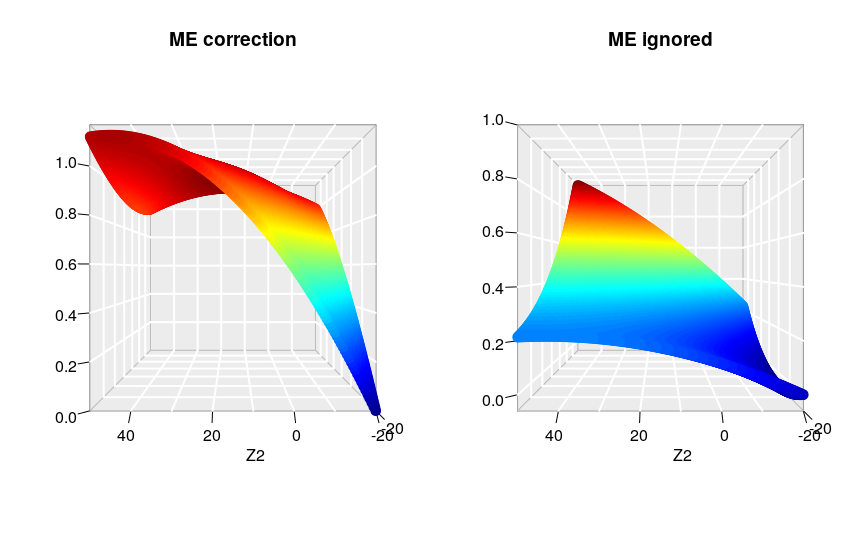}
	\vskip -1cm
	\caption{Nonparametric  estimates of $g(Z_1, Z_2)=g_1(Z_1)+g_2(Z_2) + g_3(Z_1, Z_2)$. The first column contains the estimate from our sieve rank estimator and the second column the estimate from ignoring measurement error. }
	\label{APP_3D2}	 
\end{figure}

We now discuss the plausibility of Assumptions \ref{A:IV}-\ref{A:Unobs} required for identification.
Assumptions \ref{A:IV} posits that given true beliefs $Y^*$ and relevant individual characteristics $W$, the historic return information $Z_1, Z_2$ have no impact on the mean reported belief.
Assumption \ref{A:Mon} imposes a mild restriction on the measurement error mechanism in that it requires monotonicty in the reporting of beliefs (in the conditional mean). 
Assumption \ref{A:Unobs} is satisfied as $Z_1, Z_2$ are by the experimental setup credibly fully independent of unobservables $U$.
The data consists of 1084 interviewed persons but 306 people do not respond to the question on beliefs. We removed missing values and report the summary statistics in Table \ref{table:sum}. 


We estimate functions $g_1$, $g_2$, $g_3$ with our method outlined in (\ref{CS-criterion}) and contrast the results to estimates obtained from assuming classical measurement error, i.e., from a standard additive-separable, nonparametric regression of $Y$ on $Z_1$ and $Z_2$ with the respective interaction term. We choose a B-Spline basis of degree two without interior knots for each function estimate. This choice is motivated by a 10-fold cross-validation on the model ignoring the measurement error.   

The results are presented in Figure \ref{APP_3D2}.
Accounting for the measurement error leads to a concave, symmetric effect of both treatments on the individual beliefs. When ignoring the possibility of measurement error, results are much more asymmetric, including convex marginals for the first treatment and flat parts in the surface. In contrast, our method yields that individuals learn conservatively from both treatments which is in line with the a priori economic intuition. Note that on the z-axis that estimates in both columns have been normalized to move through coordinates (-20,-20,0) and (50,50,1). Functions are evaluated on a grid ranging from -20 to 50 which corresponds to the $10\%$- and $90\%$-quantile of the marginal distributions of the treatment variables.
Summarizing, accounting for possible nonclassical measurement error in the outcome variable delivers function estimates of belief formation  that are more in line with economic intuition. 
 		 
\section{Conclusion}
This paper provides new insights on the analysis of regression models with nonclassical measurement error in the outcome variable. Our nonparametric identification result is based on intuitive assumptions involving shape restrictions on measurement error functions. This novel result builds on the equivalence of nonclassical measurement models and generalized regression models. We consider a sieve rank estimator which constructively arises from our identification result and implicitly accounts for the required shape restrictions. We establish the rate of convergence of the sieve rank estimator which is affected by a potentially ill-posed inverse problem. The proposed estimation method is easy to implement and provides numerically stable results as demonstrated in a finite sample analysis. Finally, we demonstrate the usefulness of our method in an empirical application on belief elicitation, where we find measurement error in subjective belief data to be of a nonclassical form.

\appendix
\section{Extension: Estimation with Continuous $W$} \label{Sec::ConW}

When $W$ does contain continuous variables, we can simply replace the indicator in (\ref{Estimator}) with a kernel function to account for the fact that $W_i=W_j=w$ is a null event.
Then estimation can proceed with
\begin{align}
\widehat g_w &= \arg \max_{\phi \in \mathcal{G}_{K}} \mathcal Q_n(\phi,w) \;\; \text{where}\\
\mathcal Q_n(\phi, w) &:= \sum_{1 \leq i < j \leq n} Y_i \mathcal K_s (W_{i}- w) \mathcal K_s(W_{j}- w) \mathds{1}\{\phi(Z_{i}) > \phi(Z_{j})\} \nonumber
\end{align} 
where  $K_h$ is defined as 
\begin{align*}
\mathcal K_s(W_{i}-w)=\prod_{l=1}^{d_w} \mathcal K\left(\frac{W_{l,i}-w}{s_l}\right)
\end{align*}
and $\mathcal K: \mathbb{R} \rightarrow \mathbb{R}$ is some kernel function and $s \in \mathbb{R}^{d_w}$ a vector of bandwidths. 

As we move from the original criterion of \cite{CavSherman} to the conditional version with continuous $W$ the computational complexity of the maximization problem increases. Ranking is an $O(n\log(n))$ operation whereas the weighted ranking is performed in $O(n^2)$ time.
This implies that the conditional estimation method is not scalable to large data sets and computation time increases heavily with the sample size. 

The following criterion can be used to deal with continuous W and computation time scales in $n$.
\begin{align}\label{Weighted_Crit}
\mathcal Q_n(\phi,w) & = \sum_{1 \leq i < j \leq n} \mathcal K^U_s(W_i-w) Y_i   \mathcal K^U_s(W_j-w) \mathds{1}\{\phi(Z_{i}) > \phi(Z_{j})\}   \nonumber        \\
& = \sum_{i:\;\; w-s < W_i < w+s} Y_i  \text{Rank}_s(\phi(Z_i))
\end{align}
with uniform kernel
\begin{align*}
\mathcal K^U_s(W_i-w):=\mathds{1}\{w-s < W_i < w+s\}
\end{align*}
which is again equivalent to applying the sieve rank estimator over a subsample of the data obtained by considering a window of size $2s$ around $w$. 
Weighted rank estimation is studied in \cite{shin2010} and \cite{abrevaya2011} for semiparametric and additively separable models. 
An important special case is again the setting where the function $g(\cdot ,w)$ does not vary with $w$ which is the case of $g$ is additvely separable in a function of $Z$ and $W$.
\begin{rem}\label{Rem::WeightEst}
	Assume the function $g(Z)$ does not depend on $W$. We can consider the following estimator
	\begin{align*}
	\widehat g &= \arg \max_{\phi \in \mathcal{G}_{K}} Q_n(\phi) \;\; \text{where}\\
	Q_n(\phi) &:= \sum_{1 \leq i < j \leq n} Y_i \mathcal{K}_h(W_{i}- W_{j}) \mathds{1}\{\phi(Z_{i}) > \phi(Z_{j})\}
	\end{align*} 
	In contrast to before we consider only those observations in a neighborhood around a fixed value $w$ but we choose the weights according to which distance any pair $(W_i, W_j)$ has to each other. Similar to the approach in (3.5) this is associated with increasing computational complexity as the computation time does not scale with the sample size. 
	
	We thus suggest the following strategy:
	
	First use the criterion in (\ref{Weighted_Crit}) to obtain estimates $\widehat g_w$ across different values of $w \in \textsl{supp}(W)$. Each is an estimate of $g$ as $g$ does not depend in theory on $w$, but estimation results may nevertheless vary for different $w$. 
	Second, aggregate the different estimates $\widehat g_w$ to one final estimator for $g$.
	To this end, we can follow \cite{CKK} which discuss the following two 'aggregation' procedures.
	\begin{align*}
	\widehat{g}_{LS}(z) = \arg \min_{q \in \mathbb{R}} \int_{supp(W)} \nu(w) [\widehat g(z,w)-q]^2 dw \\
	\widehat{g}_{LAD}(z) = \arg \min_{q \in \mathbb{R}} \int_{supp(W)} \nu(w) |\widehat g(z,w)-q| dw 
	\end{align*}
	where $\nu$ is some weighting function with $\int_{supp(W)} \nu(w)dw=1$. 
	
	The implementation is simple. Random draws from $\{W_i\}_{i=1}^N$ yields a set of different realizations $w$ on which to evaluate the local estimators $\widehat g_w$. The LS criterion takes the average of the local estimators, the LAD criterion takes the empirical median to aggregate to a final estimator for $g$. 
	In simulations \cite{CKK} find that the latter estimator performs best as for $w$ in the tails of the distribution of $W$ we may get erratically behaving $\widehat{g}_w$. 
\end{rem}

\subsection{Weighted Rank Estimation}
In this section we assess the performance of a weighted rank estimator for a setting as described in Remark \ref{Rem::WeightEst}. We consider the following data generating process similar to Section \ref{Sec::MC},
\begin{align*}
Y^* & =Z_1+g(Z_2)+m(W)+U\cdot W^2 \\
Y  &= h(Y^*+ W)+V\cdot |W|
\end{align*}
where $g(\cdot)=\sin(\cdot)$, $m(\cdot)=cos(\cdot)$, $W= 0.5\cdot Z_2 + 0.5\cdot U $ and the remaining variables as in Section \ref{Sec::MC} with $h$ parameterized by $a=b=0$. In this setting there is correlation between $Z_2$ and $W$. Further the measurement is additionally affected by the variable $W$. This setting is in line with Remark \ref{Rem::WeightEst} as $g$ does not vary with $W$, and we implement the procedure outlined at the end of this remark with the LAD-criterion as aggregating procedure. 

In order to calculate an estimate of $g$ for each Monte Carlo sample, we first take 50 random draws of the variable $W$, calculate $\widehat{g}_w$ by maximizing (\ref{Weighted_Crit}) for each of the 50 different realizations $w$. Finally, we aggregate the results to a final estimate by taking the sample median over the local estimates $\widehat{g}_w$.  
We vary the bandwidth parameter $\widetilde s$ across different experiments. 
The sample size is $n=1000$ and $500$ Monte Carlo replications are considered.
The following Figure \ref{Plot_Weighted} shows the results.

\begin{figure}[!h]
	
	\centering
	\includegraphics[width =14cm]{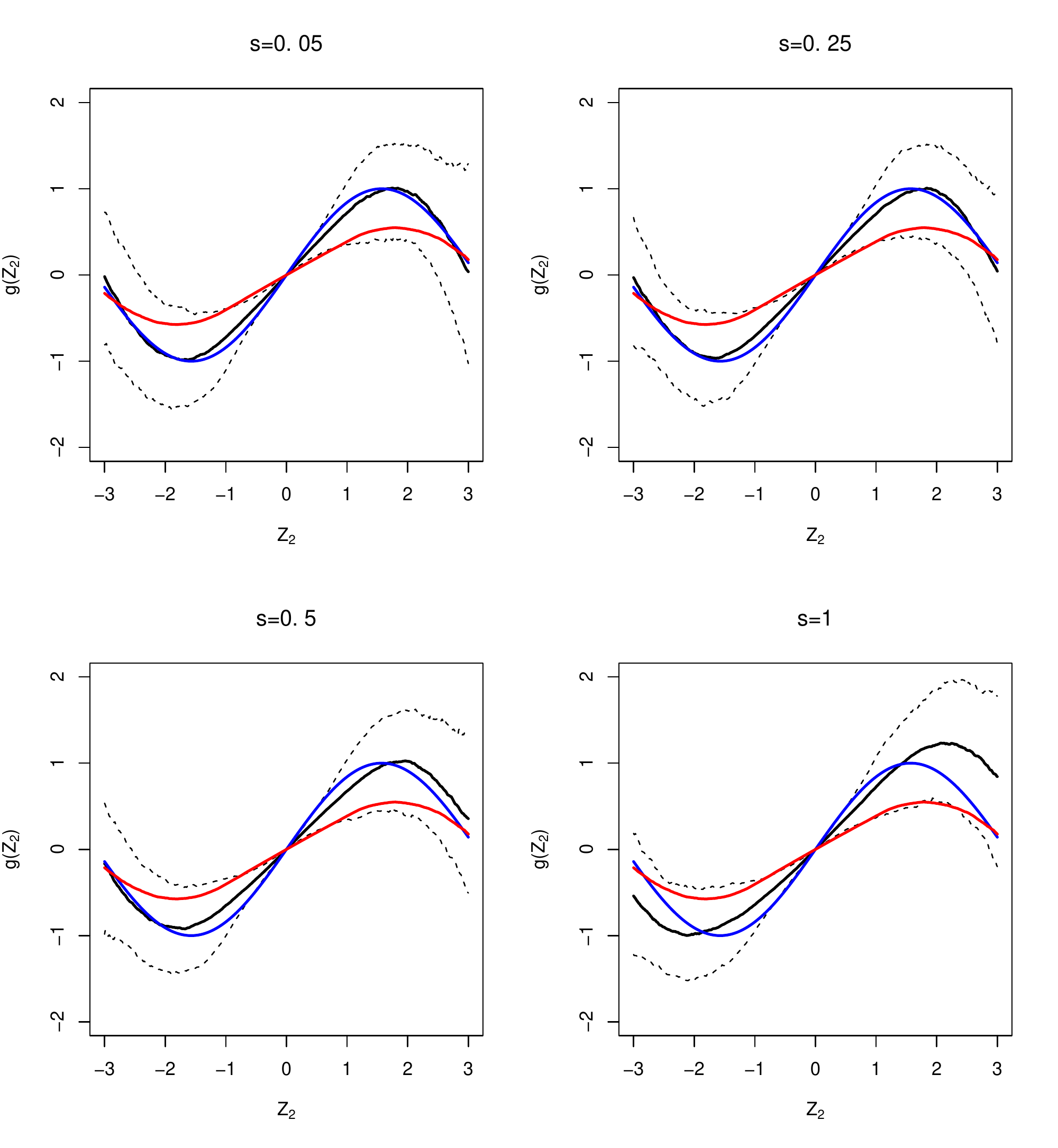}
	\caption{ The blue line is the $g(\cdot)=\sin(\cdot)$ function, the solid black line denotes the median and the dotted lines the respective 0.95 and 0.05 quantiles of the weighted sieve rank estimator over the Monte Carlo experiments. The red line is the median of series estimates of $g$ in the model $Y=Z_1+g(Z_2)+m(W)+U$. Basis functions are set as in Section \ref{Sec::MC} with $K=4$.   \label{Plot_Weighted}}
\end{figure}

If we choose $s$ reasonably small, our estimation procedure is quite close to the truth and outperforms the standard nonparametric estimator that simply ignores the measurement error. Increasing the bandwidth $s$ leads to smaller confidence bands, but considerably increases the bias of the estimate. However in this strong measurement error setting, the weighted sieve rank estimator still outperforms the estimate from ignoring the measurement error.

\section{Proofs and Technical Results}\label{Appendix:proofs}
	First, recall that $X=(Z,W)$ and that $g_w=g(\cdot, w)$. 

\begin{proof}[\textsc{Proof of Theorem \ref{thm:ident}.}]
Proof of \eqref{eq:transform}. The exclusion restriction captured in Assumption \ref{A:IV} implies
	\begin{align}\label{Impl.1-3}
	\Ex[Y|X=x] 
	&= \Ex[h(Y^*,W)\;|\; Z=z, W=w] \nonumber \\
	&= \Ex[h(g(Z, W)+U, W)\;| \; Z=z, W=w] \nonumber\\
	&= \Ex[h(g(z, W)+U, W) \;|\; W=w],
		\end{align}
	where the last equation is due  to the conditional exogeneity  imposed in Assumption \ref{A:Unobs}. The results follows from strict montonicity of $H(g_w(z), w) =\Ex[h(g(z, W)+U, W) \;|\; W=w]$ in its first argument, which is due to  Assumption \ref{A:Mon} and Assumption \ref{regIdent} (ii). 
	
Proof of \eqref{eq:transform:max}. 	By the law of iterated expectations, the criterion function $Q(\phi,w)=\Ex[Y_1\mathds{1}\{\phi(X_1) > \phi(X_2)\}|W_1=W_2=w]$ can be rewritten as
	\begin{equation*}
    Q(\phi,w)=\frac{1}{2}\Ex[H(g(X_1), W_1)\mathds{1}\{ \phi(X_1) > \phi(X_2)\} 
     + H(g(X_2), W_2)\mathds{1}\{ \phi(X_1) < \phi(X_2) \}  \;|\; W_1=W_2=w],
	\end{equation*} 
using $\Ex[Y|X]=H(g(X),W)$ by equation \eqref{eq:transform}. Under Assumption \ref{A:Mon}, we may consider the case that  holds with $h(\cdot,w)$ weakly monotonically increasing, without loss of generality.
	Now the function $g_w$ is a maximizer of $Q(\cdot,w)$, which follows by
	\begin{equation*}
	Q(g_w,w) = \frac{1}{2} \Ex[\max\{H(g(X_1),W_1), H(g(X_2),W_2)  \}  \;|\; W_1=W_2=w]
	\end{equation*}
	and using monotonicity of $H$ in its first argument. In particular,  $m \circ g_w$ is a maximizer of $Q(\cdot,w)$ for any strictly increasing function $m$ (here $\circ$ denotes function composition).

	
	It remains to show that $g_w$ is a unique maximizer up to strictly increasing transformations.
	Specifically, we show that for any function $\widetilde g_w \neq m \circ g_w$ for an arbitrary strictly monotonic transformation $m$ we have that $Q(\widetilde g_w,w) <Q(g_w, w)$. 
	To do so, consider some arbitrary function $\phi \in \mathcal{G}$ that is not a strictly monotonic transformation of $g_w$. Therefore, there exist $z', z''\in\textsl{supp}(Z)$ such that
$g_w (z') < g_w(z'')$ and $	\phi(z') > \phi(z'')$. 
	By \eqref{eq:transform}, $H(\cdot,w)$ is strictly monotonic and it holds for every $w$ that
	\begin{equation*}
	H(g_w(z'),w) < H(g_w(z''),w).
	\end{equation*}
	By continuity of the functions following Assumption \ref{regIdent} (i) the above inequalities hold in neighborhoods $B_1$ around $z'$ and $B_2$ around $z''$, respectively. By Assumption \ref{regIdent} (iii) these neighborhoods have a strictly positive probability measure.
	This implies
	\begin{align*}
	 Q(g_w,w)-Q(\phi ,w)
	 \geq & \frac{1}{2}\Ex[H(g_w(Z_1), W_1) - H(g_w(Z_2), W_2) | Z_1, Z_2 \in B_1 \times B_2, W_1=W_2=w]& \\
	 &  \times \PP(Z_1, Z_2 \in B_1 \times B_2\;|\; W_1=W_2=w) \\
	 >&  0.
	\end{align*}
Thus, $Q(\cdot,w)$ is only maximized by $g_w$ and strictly monotonic transformations of it. Hence, $g_w$ is identified up to a strictly monotonic transformation.
\end{proof}
\begin{proof}[\textsc{Proof of Corollary \ref{coro:ident}.}] 
Under Assumption \ref{shapecons} (i) any candidate regression function $\widetilde g_w(Z)=\widetilde m_w(Z_1)+\widetilde l_w(Z_{-1})$ must satisfy 
\begin{align*}
\widetilde g_w(Z)&= M_w(g_w(Z))=M_w(m_w(Z_1)+l_w(Z_{-1})=\widetilde m_w(Z_1)+\widetilde l_w(Z_{-1})
\end{align*}
for a strictly monotonic function $M_w$. Thus $M_w$ must be linear and $g_w$ is identified up to location and scale transformation. Indeed, given linear and  strictly monotonic transformations, $g_w$ is the only maximizer of $Q(\cdot,w)$. Under Assumption \ref{shapecons} (ii) we have that $g_{w}(z_1)=\Ex[Y|Z=z_1, W=w]$ and $g_{w}(z_2)=\Ex[Y|Z= z_2, W=w]$ and fixing the parameter space to move through both points leads to $g_w$ being the unique maximizer of $Q(\cdot,w)$ over $\mathcal{G}$ and thus $g_w$ is point identified.
\end{proof}

\begin{proof}[\textsc{Proof of Lemma \ref{Lem::AddSep}.}]
	Let $Z_1, Z_2$ be independent copies of $Z$.
	Consider the additive separable case $g(Z_1)=Z_{11} + \widetilde g(Z_{12})$ with bivariate $Z_1=(Z_{11}, Z_{12})$.
	Analogously we denote $\phi(Z_1)=Z_{11} + \widetilde \phi(Z_{12})$. 
	The following holds for the criterion $\mathcal Q$ 
	\begin{align*}
	|\mathcal Q(\phi) |
	=&  \Ex[Y_1(\mathds{1}\{Z_{11} + \widetilde g(Z_{12})  >  g(Z_{2})\} - \Ex[Y_1\mathds{1}\{Z_{11} + \widetilde \phi(Z_{12})  >  \phi(Z_{2})\}]   \\
	=&  \Ex[Y_1 (F_{Z_{21} | Z_{22}}( \phi(Z_{1})- \widetilde \phi(Z_{21})) -F_{Z_{21} | Z_{22}}( g(Z_{1})- \widetilde g(Z_{21})) )]    ,
	\end{align*}
	as $g$ is the maximizer of $\mathcal{Q}$ and with the second equation due to the law of iterated expectation.
	Using a second-order Taylor decomposition with directional derivatives yields for all $\phi$ in a neighborhood around $g$
	\begin{align*}
	|\mathcal Q(\phi) | = &  Q_g(\phi -g)  + \underbrace{\Ex[Y_1 f^{''}_{Z_{21} | Z_{22}}(\xi)(\widetilde\phi(Z_{12}) - \widetilde g(Z_{12}) +    \widetilde g(Z_{22}) - \widetilde \phi (Z_{22}) )^3 ]}_{=R},
\end{align*}
where $\xi$ is some intermediate variable\footnote{More precisely $\xi=g(Z_{1})- \widetilde g(Z_{22})+ s[\phi(Z_1)-\widetilde \phi(Z_{21}) + \widetilde g(Z_{21})-g(Z_{1})]$ for some $s\in (0,1)$.} and $Q_g$ denotes the directional derivative of $\mathcal Q$ at $g$ which is given by
\begin{align*}
 Q_g(\phi-g)= &\Ex[Y_1 f'_{Z_{21} | Z_{22}}(g(Z_{1})- \widetilde g(Z_{21}))(\widetilde\phi(Z_{12}) - \widetilde g(Z_{12}) +    \widetilde g(Z_{22}) - \widetilde \phi (Z_{22}) )^2 ].
	\end{align*}
	 Applying the Cauchy-Schwarz inequality to $Q_g(\phi-g)$ shows that $Q_g$ is weaker than the $L^2$-norm.
	Further, the remainder term $R$ satisfies
	\begin{equation*}
	|R| \leq \Ex \left[\left|\frac{f^{''}_{Z_{21} | Z_{22}}(\xi)}{f'_{Z_{21} | Z_{22}}(g(Z_{1})- \widetilde g(Z_{21}))}(\widetilde\phi(Z_{12}) - \widetilde g(Z_{12}) +    \widetilde g(Z_{22}) - \widetilde \phi (Z_{22}) )\right| \right] \cdot Q_g(\phi-g)
	\end{equation*}
	and thus the tangential cone condition in Assumption \ref{A:Rate} (iv) is satisfied if the first factor on the right hand side is bounded between 0 and 1. The lower bound holds directly and the upper bound is easily satisfied if the $\delta-$ neighborhood around $g$ is chosen sufficiently small and derivatives of the density are bounded away from zero and infinity, as is condition.
\end{proof}
For the proof of the next results, we require some additional notation to deal with the Hoeffding decomposition of U-statistics, specific function spaces and their respective envelope functions.

We introduce the empirical criterion $\mathcal Q_n(\phi)$ that can be denoted as
\begin{align*}
 \mathcal Q_n(\phi) &  = \frac{2}{n(n-1)}\sum_{1 \leq i < j \leq n} \Gamma(S_i, S_j, \phi)  
\end{align*}
where $S_i=(Y_i, Z_i)$ and which is a second order U-statistic with kernel
\begin{align*}
\Gamma(S_i, S_j, \phi) =  Y_i\big(\mathds{1}\{\phi(Z_i)>\phi(Z_j)\}- \mathds{1}\{g(Z_i)>g(Z_j)\}\big)
\end{align*}
indexed by $\phi \in \mathcal{G}_K$ making it a second-order U-process. Note that $\mathcal{Q}_n$ is centered here which does not affect the optimization. Using the kernel notation, the criterion function $\mathcal{Q}$ given in \eqref{def:crit:centered} satisfies $\mathcal{Q}(\phi)= \Ex[\Gamma(S_i, S_j, \phi)]$.

For the asymptotic analysis we make use of the Hoeffding decomposition of a U-statistic (see e.g. \cite{vanVaart98})
\begin{align}\label{Hoeffding}
\mathcal{Q}_n(\phi) =   \mathcal{Q}(\phi) + \nu_n(\phi) +  \xi_n(\phi) 
\end{align}
with short hand notations
\begin{align*}
\nu_n(\phi) &:= \frac{1}{n}\sum_{i=1}^{n} \nu(S_i, \phi),\\
\nu(S_i, \phi) &:=  \Ex[\Gamma(S_i, S_j, \phi)| S_i] + \Ex[\Gamma(S_j, S_i, \phi)|S_i]- 2\Ex[\Gamma(S_i, S_j,\phi)],  \\
\xi_n(\phi) &:= \frac{2}{n(n-1)} \sum_{1 \leq i < j \leq n} \xi(S_i, S_j, \phi), \\
\xi(S_i, S_j, \phi) &:= \Gamma(S_i,S_j, \phi) - \Ex[\Gamma(S_i,S_j, \phi)| S_i] - \Ex[\Gamma(S_i, S_j, \phi)|S_j] + \Ex[\Gamma(S_i, S_j,\phi)].
\end{align*}
This decomposition is frequently deviced in the rank estimation literature to obtain asymptotic results, see e.g. \cite{sherman93}. The first summand in the decomposition is a smooth function of the parameter $\phi$, $\nu_n$ is an empirical process and $\xi_n$ a degenerate U-process, both indexed by the function space $\mathcal{G}_K$.

Further, we define the function classes $\mathcal{F}_{\nu,K}=\{ \nu(\cdot, \phi): \phi \in \mathcal{G}^\delta_K\}$ and  $\mathcal{F}_{\xi,K}=\{ \xi(\cdot, \cdot, \phi): \phi \in \mathcal{G}_K^\delta\}$. Let $\overline F_\nu$ and $\overline F_\xi$ denote respective envelope functions .
The envelope function is defined as any function satisfying $|\nu(\cdot,\phi)|\leq \overline F_\nu(\cdot)$.
In this setting, $\overline F_\nu(S_i) = |Y_i| + 3 \Ex[|Y_i|]$, since
\begin{align*}
|\nu(S_i, \phi)| = &| Y_i\Ex[\mathds{1}\{\phi(Z_i) > \phi(Z_j)\} - \mathds{1}\{g(Z_i) > g(Z_j)\}|Z_i]\\
& +\Ex[Y_j \left(\mathds{1}\{\phi(Z_j) > \phi(Z_i)\} - \mathds{1}\{g(Z_j) > g(Z_i)\} \right)|Z_i] \\
& - 2\Ex[ Y_i \left(\mathds{1}\{\phi(Z_i) > \phi(Z_j)\} - \mathds{1}\{g(Z_i) > g(Z_j)\} \right) ]| \\
\leq & |Y_i| + 3\Ex[|Y_i|],
\end{align*}
where $\mynorm*{\overline F_\nu}_{L^2(S)} \leq \sqrt{4\Ex[Y^2]} =: C_\nu$. 
In addition we have $\overline F_{\xi}(S_i, S_j)=2|Y_i|+2\Ex[|Y_i|]$ as
\begin{align*}
|\xi(S_i, S_j, \phi)| = & | Y_i(\mathds{1}\{\phi(Z_i) > \phi(Z_j)\} - \mathds{1}\{g(Z_i) > g(Z_j)\}) \\
-&  Y_i\Ex[\mathds{1}\{\phi(Z_i) > \phi(Z_j)\} - \mathds{1}\{g(Z_i) > g(Z_j)\} |Z_i ] \\
-&  \Ex[Y_i(\mathds{1}\{\phi(Z_i) > \phi(Z_j)\} - \mathds{1}\{g(Z_i) > g(Z_j)\}) |Z_j ] \\
+&  \Ex[Y_i(\mathds{1}\{\phi(Z_i) > \phi(Z_j)\} - \mathds{1}\{g(Z_i) > g(Z_j)\})] |
\end{align*}
and $\mynorm*{\overline F}_{L^2(S)} \leq \sqrt{12\Ex[Y^2]} =: C_\eta$. By Assumption \ref{A:Rate} (iii) we have $C_\nu, C_\eta< \infty$.
Ultimately, we define the bracketing integral $J_{[]}$ of the space $\mathcal{F}_{\nu, K}$  
\begin{align*}
J_{[]}(1,\mathcal{F}_{\nu,K},  L^2(S))= \int_{0}^{1} \sqrt{1+ \log N_{[]} (\epsilon \cdot \mynorm*{\overline F_\nu}_{L^2(S)},\mathcal{F}_{\nu,K}, L^2(S)) } d\epsilon.
\end{align*}
and analogously for $\mathcal{F}_{\xi, K}$.

\begin{proof}[\textsc{Proof of Theorem \ref{thm::rate}.}]
We begin by noting that consistency of $\widehat{g}$ in the $L^2$-norm follows from Lemma \ref{App::Lemma2}.
Due to the consistency result in  Lemma \ref{App::Lemma2}, we may restrict the function spaces to a local neighborhood around $g$, i.e. we define the space $\mathcal{G}^\delta_K = \{\phi \in \mathcal{G}_K : \mynorm*{\phi-g}_{L^2(Z)}< \delta\}$ and assume that $\widehat g \in \mathcal{G}^\delta_K $. Further we introduce the space $\mathcal{G}^{\delta,r_n}_K =\{ \phi \in \mathcal{G}^\delta_K :  Q_g(\phi-g) > Mr_n  \}$ where $M>0$.
It holds that
\begin{align*}
 \PP\big( Q_g(\widehat{g}-g) &\geq Mr_n\big)
\leq  \PP \left( \sup_{\phi\in\mathcal{G}^{\delta,r_n}_K} \mathcal Q_n(\phi) \geq  \mathcal Q_n(\Pi_Kg) \right)\\
\leq & \PP \left( \sup_{\phi\in\mathcal{G}^{\delta,r_n}_K} \mathcal Q(\phi) + \nu_n(\phi) + \xi_n ( \phi) \ \geq \mathcal Q(\Pi_Kg) + \nu_n(\Pi_Kg) + \xi_n (\Pi_Kg) \right),
\end{align*}
by applying the Hoeffding decomposition (\ref{Hoeffding}).
Due to Assumption \ref{A:Rate} (iv) we have local equivalence of $|\mathcal Q(\cdot)|$ and $Q_g(\cdot)$. Since $\mathcal Q(\cdot)$ is negative and thus $|\mathcal Q(\cdot)|=-\mathcal Q (\cdot)$ it follows that
{\small \begin{flalign*}
  \PP \big( Q_g(\widehat{g}-g) \geq Mr_n \big) &\leq  \PP \left( \sup_{\phi\in\mathcal{G}^{\delta,r_n}_K} \Big(\mathcal Q(\phi) +  \nu_n(\phi)-\nu_n(\Pi_Kg) + \xi_n (\phi) -\xi_n (\Pi_Kg)) \Big)\geq - \eta Q_g(\Pi_Kg-g) \right) &\\
\leq & \PP \left(\sup_{\phi\in\mathcal{G}^{\delta,r_n}_K} \Big(\nu_n(\phi)-\nu_n(\Pi_Kg) + \xi_n (\phi) -\xi_n (\Pi_Kg)+\eta Q_g(\Pi_Kg-g) \Big)\geq \inf_{\phi\in\mathcal{G}^{\delta,r_n}_K} |\mathcal Q(\phi)| \right) \\
\leq & \PP \left(\sup_{\phi\in\mathcal{G}^\delta_K} \nu_n(\phi)-\nu_n(\Pi_Kg)  + \sup_{\phi\in\mathcal{G}^\delta_K} \xi_n (\phi) -\xi_n (\Pi_Kg) + \eta Q_g(\Pi_Kg-g) \geq C_2Mr_n  \right),
\end{flalign*}}%
where it remains to study the asymptotic behavior of each summand in the last line separately. 
Note that both summands on the left hand-side are positive, hence if $\sup_{\mathcal{G}^\delta_K}  \nu_n(\phi)$ is bounded in probability so is $ \nu_n(\Pi_Kg)$ and similarly for $\xi_n$.

First we study the asymptotic behavior of the empirical process part $\sup_{\phi \in\mathcal G^\delta_K}  \nu_n(\phi)$. Recall the definition $\mathcal{F}_{\nu,K}=\{ \nu(\cdot, \phi): \phi \in \mathcal{G}^\delta_K\}$ with envelope $\overline F_\nu$. 
By applying the last display of Theorem 2.14.2 of \cite{Vaart2000} we can conclude that 
\begin{align*}
\Ex \Big| \sup_{\phi \in\mathcal G^\delta_K} \nu_n(\phi)\Big| = \Ex \Big|\sup_{\nu \in \mathcal{F}_{\nu,K}} \frac{1}{n}\sum_{i=1}^{n} \nu(S_i) \Big| \leq J_{[]}(1,\mathcal{F}_{\nu,K}, L^2(S))\cdot \mynorm*{\overline F_\nu}_{L^2(S)}\cdot n^{-1/2}
\end{align*}
where $\mynorm*{\overline F_\nu}_{L^2(S)} \leq \mynorm*{\overline F_\nu}_{L_\infty(S)} \leq C_\nu< \infty$. 
By Lemma \ref{App::Lemma1} (i) and (ii) we have
\begin{align*}
\log N_{[]} (\epsilon \cdot \mynorm*{\overline F_\nu}_{L_\infty(S)}, \mathcal{F}_{\nu,K},  L_\infty(S)) \leq c_0 K \log(1/\epsilon \cdot C_\nu^{-1} )
\end{align*}
and ultimately we obtain $J_{[]}(1, \mathcal{F}_{\nu,K}, L_\infty(S))=O(\sqrt{K})$ and by Markov's inequality $\sup_{\phi \in \mathcal{G}^\delta_K} \nu_n(\phi) = O_p(\sqrt{K/n})$.

It remains to analyze the convergence rate of the degenerate U-process $\sup_{\phi \in \mathcal{G}^\delta_K} \xi_n(\phi)$. Similar to Lemma A.1 in \cite{ClemLugosi08} we can make use of the following equality for second-order U-statistics 
\begin{align}\label{ClemLug}
\frac{1}{n(n-1)} \sum_{i \neq j} \xi(S_i, S_j, \phi) = \frac{1}{n!} \sum_{\pi} \frac{1}{\lfloor n/2 \rfloor} \sum_{i=1}^{\lfloor n/2 \rfloor} \xi(S_i, S_{\lfloor n/2 \rfloor + i}, \phi)
\end{align}
where $\pi$ is short-hand for all permutations of  $\{1,\dots, n\}$. Then applying the triangle inequality to (\ref{ClemLug}) leads to 
\begin{equation}\label{UStatIneq}
\Ex \left[ \Big| \sup_{\phi \in \mathcal{G}^\delta_K}  \frac{1}{n(n-1)} \sum_{i \neq j} \xi(S_i, S_j, \phi) \Big| \right] \leq \Ex \left[ \Big |\sup_{\phi \in \mathcal{G}^\delta_K}  \frac{1}{\lfloor n/2 \rfloor} \sum_{i=1}^{\lfloor n/2 \rfloor} \xi(S_i, S_{\lfloor n/2 \rfloor + i}, \phi) \Big | \right]
\end{equation}
from which we can conclude that for obtaining the convergence rate of the degenerate U-process on the left-hand side of (\ref{UStatIneq}) it is sufficient to analyze the convergence rate of an empirical process with kernel $\xi$ indexed by the function $\mathcal{G}^\delta_K$.

The kernel $\xi$ contains non-smooth indicator functions so we cannot apply the exact same reasoning we used earlier to derive a bound for $\nu_n$, as $\xi(S_i, S_j, \phi)$ is not continuous in $\phi$.  
However we can use the fact that $\xi(\cdot, \cdot, \phi)$ belongs to a VC- subgraph family and we can thus derive the complexity bound in Lemma \ref{App::Lemma1} (iii).

Recall the definition $\mathcal{F}_{\xi, K}=\{ \xi(\cdot, \cdot, \phi): \phi \in \mathcal{G}_K^\delta\}$ and the associated envelope function $\overline F_\xi$.
Now we apply Theorem 2.14.1 of \cite{Vaart2000}
\begin{equation*}
\Ex \left[ \Big |\sup_{\phi \in \mathcal{G}^\delta_K}  \frac{1}{\lfloor n/2 \rfloor} \sum_{i=1}^{\lfloor n/2 \rfloor} \xi(S_i, S_{\lfloor n/2 \rfloor + i}, \phi) \Big | \right] \leq J_{[]}(1,\mathcal F_{\xi, K}, L^2(S))\mynorm*{\overline F_\xi}_{L^2(S)}\lfloor (n/2)\rfloor^{-1/2}
\end{equation*}
Applying Lemma \ref{App::Lemma1} (iii) we obtain the bound
\begin{align*}
J_{[]}(1, \mathcal{F}_{\xi, K}, L^2(S)) \leq  \int_{0}^{1} \sqrt{1+ c_1+c_2K\log(1/\epsilon) }d\epsilon = O(\sqrt{K})
\end{align*}
and by Markov's inequality that $\sup_{\phi \in \mathcal{G}^\delta_K}  \xi_n(\phi)=O_p(\sqrt{K/n})$.
Finally, we can conclude that
\begin{equation*}
\PP \left( Q_g(\widehat g-g) \geq Mr_n \right)
\leq \PP \left(\sup_{\phi \in \mathcal{G}^\delta_K}  \nu_n(\phi)  + \sup_{\phi \in \mathcal{G}^\delta_K}  \xi_n(\phi) + Q_g(\Pi_Kg-g) \geq C_2Mr_n \right).
\end{equation*}
with $\sup_{\phi \in \mathcal{G}^\delta_K}  \nu_n(\phi)=O_p(\sqrt{K/n})$ and $\sup_{\phi \in \mathcal{G}^\delta_K}  \xi_n(\phi)=O_p(\sqrt{K/n})$.
Consequently,  choosing $r_n=\max\{ \sqrt{K/n},Q_g(\Pi_Kg-g)\}$ we see that the right hand side probability converges to zero as $M \to \infty$. Thus $Q_g(\widehat{g}-g) =O_p(r_n)$.  By the definition of the sieve measure of ill-posedness $\tau_K$ we obtain
\begin{equation*}
\mynorm*{\widehat{g} - g}_{L^2(Z)}  \leq \tau_K Q_g(\widehat{g} - g) 
                            \leq \tau_K O_p\left( \max\{ \sqrt{K/n}, Q_g(\Pi_Kg-g)\} \right) = O_p\left( \tau_K \sqrt{K/n}, \mynorm*{\Pi_Kg-g}_{L^2(Z)}\right) 
\end{equation*}
which concludes the proof. 
\end{proof}

\begin{lem}\label{App::Lemma1}
Under Assumption \ref{A:Rate} it holds that\\
(i) $\sup_{\mynorm*{\phi-g}_\infty \leq \delta} |\nu(S_i,\phi)|\leq M_1(S_i)\cdot \delta$ with $\Ex[M_1(S_i)] < \infty $, 	\\
(ii) $\log N_{[]} (\epsilon,\mathcal{F}_{\nu,K}, L_\infty(S)) \leq  c_0K\log(1/\epsilon) $ for some positive constant $c_0$, \\
(iii)  $\log N(\epsilon,\mathcal{F}_{\xi,K}, L^2(S)) \leq c_1+c_2K\log(1/\epsilon)$, for positive constants $c_1,c_2$.
\end{lem} 
\begin{proof}[\textsc{Proof of Lemma \ref{App::Lemma1}.}] 
Proof of part (i). It holds that
\begin{align*}
\nu(S_i, \phi) = & Y_i\Ex[\mathds{1}\{\phi(Z_i) > \phi(Z_j)\} - \mathds{1}\{g(Z_i) > g(Z_j)\}|Z_i]\\
& +\Ex[Y_j \left(\mathds{1}\{\phi(Z_j) > \phi(Z_i)\} - \mathds{1}\{g(Z_j) > g(Z_i)\} \right)|Z_i] \\
& - 2\Ex[ Y_i \left(\mathds{1}\{\phi(Z_i) > \phi(Z_j)\} - \mathds{1}\{g(Z_i) > g(Z_j)\} \right) ]
\end{align*}
We make use of the fact that as $ \mynorm*{\phi-g}_\infty \leq \delta$ and thus $g(z) - \delta \leq \phi(z)\leq g(z)+\delta$ for any $z$ in the support of $Z$. Following \cite{ChenShen2003} (p. 1599-1600) we have that
\begin{equation*}
\sup_{\mynorm*{\phi-g}_\infty \leq \delta}|\mathds{1}\{\phi(Z_j) < \phi(Z_i) \} - \mathds{1}\{g(Z_j) < g(Z_i)\} |
\leq  |\mathds{1}\{g(Z_j) < \phi(Z_i)+\delta \} - \mathds{1}\{g(Z_j) < g(Z_i) - \delta \} | 
\end{equation*}
and thus 
\begin{align*}
|\nu(S_i, \phi)| \leq  &  |Y_i|\cdot |F_{g(Z)}( \phi(Z_i)+\delta) - F_{g(Z)}( g(Z_i) -\delta) |  \\
                      & + |\Ex[Y_i|Z_i]|\cdot |F_{g(Z)}( \phi(Z_i)+\delta) - F_{g(Z)}( g(Z_i) -\delta) | \\
                      &+ |\Ex[Y_i]|\cdot \Ex[|F_{g(Z)}( \phi(Z_i)+\delta) - F_{g(Z)}( g(Z_i) -\delta) |  ] \\
                 \leq & (|Y_i| + |\Ex[Y_i|Z_i]| + |\Ex[Y_i]|) \cdot 3\delta    
\end{align*}
where the last inequality follows from Assumption \ref{A:Rate} (v), the Lipschitz continuity for the cdf of $g(Z)$. Define $M_1(S_i) = |Y_i| + |\Ex[Y_i|Z_i]| + |\Ex[Y_i]| $. From Assumption \ref{A:Rate} (iii) follows that $\Ex[M_1(S_i)]< \infty$ which concludes the argument.

We continue with the proof of  part (ii).
By Lemma \ref{App::Lemma1} (i) we have
\begin{align*}
\log N_{[]} (\epsilon, \mathcal{F}_{\nu,K},  L_\infty(S)) \leq \log N_{[]} (\epsilon, \mathcal{G}_K,  L_\infty(Z)) \leq cK \log(1/\epsilon)
\end{align*} 
where both inequalities are due to \cite{Chen07} (pp. 5595 and 5601).

We conclude with the proof of part (iii). We make use of the decomposition $\xi(S_i, S_j, \phi)  = \xi_1(S_i, S_j, \phi) + \xi_2(S_i, S_j, \phi)$ where $\xi_1(S_i, S_j, \phi)  = \Gamma(S_i,S_j, \phi) $ and 
\begin{align*}
\xi_2(S_i, S_j, \phi)  = - \Ex[\Gamma(S_i,S_j, \phi)| S_i] - \Ex[\Gamma(S_i, S_j, \phi)|S_j] + \Ex[\Gamma(S_i, S_j,\phi)].
\end{align*}
Following for instance \cite[Lemma 16]{NolanPollard} we conclude
\begin{align*}
\log N(\epsilon,\mathcal{F}_{\xi,K}, L^2(S)) \leq \log N(\epsilon,\mathcal{F}_{\xi_1,K}, L^2(S)) + \log N(\epsilon,\mathcal{F}_{\xi_2,K}, L^2(S)).
\end{align*}
Similar to the proof of part (ii) of Lemma \ref{App::Lemma1} we obtain
$\log N(\epsilon,\mathcal{F}_{\xi_2,K}, L^2(S))\leq cK \log(1/\epsilon) $ for some constant $c$.
Below, we follow Chapter 5 of \cite{sherman93} to establish that $ \mathcal{F}_{\xi_1,K}$ belongs to a VC-subgraph class.
To this end define the subgraph 
\begin{flalign*}
\text{subgraph}\left(\xi_1(\cdot, \cdot, \phi) \right)=&\{(s_i,s_j,t) \in \textsl{supp}(S)^2\times \mathbb{R}: 0<t< y_i[\mathds{1}\{\phi(z_i) > \phi(z_j) \} - \mathds{1}\{g(z_i) < g(z_j)\} ] \}& \\
 = &\{y_i > 0\}\{\phi(z_i)- \phi(z_j) > 0  \}\{t > 0 \}\{t < \overline F_{\xi_1}(z_i, z_j) \} \{g(z_i)-g(z_j)  < 0 \} \\
 &\cup \{y_i < 0\}\{\phi(z_i)- \phi(z_j) < 0  \}\{t > 0 \}\{t < \overline F_{\xi,1}(z_i, z_j) \} \{g(z_i)-g(z_j)  > 0 \}
\end{flalign*}
and introduce the function 
\begin{align*}
m(t, s_1, s_2; \gamma_1, \gamma_2, \pi_1, \pi_2):=\gamma_1t+ \gamma_2y_1+(g(z_1), p^K(z_2))'\pi_1 + (g(z_2), p^K(z_2))'\pi_2 
\end{align*} 
with the associated function space
\begin{align*}
\mathcal{M}=\{m(\cdot, \cdot, \cdot; \gamma_1, \gamma_2, \pi_1, \pi_2): \gamma_1\in \mathbb{R}, \gamma_2\in \mathbb{R}, \pi_1\in \mathbb{R}^{K+1}, \pi_2\in \mathbb{R}^{K+1}\}.
\end{align*} 
Note that $\mathcal{M}$ is a finite vector space of dimension $2(K+2)$ and the subgraph can be written as
\begin{align}\label{subgraph_decomp}
\text{subgraph}\left(\xi_1(\cdot, \cdot, \phi) \right) & =\bigcup_{i=1}^{10} \{m_i > 0\}
\end{align}
with functions $m_i \in \mathcal{M}$ for any $i=1,\dots, 10$. 
Following e.g. Lemma 2.4 and 2.5 in \cite{pakespollard89} it can be established that $\text{subgraph}\left(\xi_1(\cdot, \cdot, \phi) \right)$ belongs to a VC-class of sets and thus the space $\mathcal{F}_{\xi_1}$ is a VC-class of functions.
To bound the complexity of the space we require the VC-index of $\mathcal{F}_{\xi_1}$ which we denote as $V(\mathcal{F}_{\xi_1})=V(\text{subgraph}(\xi_1))$.
 
From \cite[ Lemma 18]{Pollard84} it follows that $V(\{m_i > 0\})\leq 2(K+2)$. Applying in \cite[Theorem 1.1]{vdVaart_VCdim} to (\ref{subgraph_decomp}) then leads to $V(\text{subgraph}(\xi_1))\lesssim 2(K+2)$, so the VC-index of the space $\mathcal{F}_{\xi_1}$ increases with the same order as the sieve dimension $K$. 
Now applying \cite[Theorem 2.6.7]{vanVaart98} yields
\begin{align*}
\log N(\epsilon,\mathcal{F}_{\xi_1,K}, L^2(S)) &\leq \log( C\cdot V(\mathcal{F}_{\xi_1}) (16e)^{V(\mathcal{F}_{\xi_1})}(1/\epsilon)^{2V(\mathcal{F}_{\xi_1})-2}) \\
& = \log(C)+ \log(2(K+2)) + 2(K+2)\log(16e)  + 2(K+2)\log(1/\epsilon)
\end{align*}
and together with $\log N(\epsilon,\mathcal{F}_{\xi_2,K}, L^2(S))\leq cK \log(1/\epsilon) $ the stated result follows. 
\end{proof}

\begin{lem}\label{App::Lemma2}
	Under Assumptions \ref{A:IV}--\ref{A:Rate} it holds that $\mynorm{\widehat g - g}_{L^2(Z)} = o_p(1)$. 
\end{lem} 
\begin{proof}[\textsc{Proof of Lemma \ref{App::Lemma2}.}] 
		We need to check the conditions in Lemma A.2 of \cite{ChenPouzo12}. 
		In their notation $\overline Q_n = \mathcal{Q}$ and
		\begin{align*}
		g_0(k, n, \epsilon)=\inf_{\phi \in \mathcal{G}_K: \mynorm{\phi- g}_{L^2(Z)} \geq \epsilon} |\mathcal{Q}(\phi)|
		\end{align*}
		Their condition a is thus satisfied and $g_0(n,k, \epsilon) > 0 $ by the identification result in Theorem \ref{thm:ident}. 
	    Condition b holds by Assumption \ref{A:Rate} (ii) and the fact that for large enough $K$ the following holds
       \begin{align*}
       |\mathcal Q(\Pi_Kg) -\mathcal Q(g)| \lesssim Q_g(\Pi_Kg - g) \lesssim \tau^{-1}_K \mynorm{\Pi_Kg -g}_{L^2(Z)},
       \end{align*} 
	    and thus  $\mathcal Q(\Pi_Kg) - \mathcal Q(g) = o(1)$.
	    Next, Condition c is implicitly assumed to hold and it remains to check condition d which translates as 
		\begin{align*}
		\frac{\max\{ |\mathcal Q(\Pi_Kg) - \mathcal Q(g)| , \sup_{\phi \in \mathcal{G}_K} |\mathcal{Q}_n(\phi) -\mathcal Q(\phi)|    \}}{g_0(n,k, \epsilon)} =o(1).
		\end{align*}
		Analogous to the empirical process result from (\ref{ClemLug}) and (\ref{UStatIneq}) and the subsequent proceedings, it holds that		$\sup_{\phi \in \mathcal{G}_K} |\mathcal{Q}_n(\phi) -\mathcal Q(\phi)| \lesssim \sqrt{K/n}$.
				Then ultimately consider that for any $\epsilon>0$ there is some $\epsilon^*>0$ that is  sufficiently small such that the local equivalence relation in Assumption \ref{A:Rate} (iv) is valid and we can conclude  
		{\small \begin{equation*}
		g_0(k,n, \epsilon) = \inf_{\mathcal{G}_K: \mynorm{\phi -g}_{L^2(Z)}\geq \epsilon} |\mathcal{Q}(\phi)|  \geq \inf_{\mathcal{G}_K: \mynorm{\phi -g}_{L^2(Z)}\geq \epsilon^*} Q_g(\phi-g) \geq \inf_{\mathcal{G}_K: \mynorm{\phi -g}_{L^2(Z)}\geq \epsilon^*} \tau^{-1}_K \mynorm{\phi -g}_{L^2(Z)}  \geq \tau^{-1}_K \epsilon^*.
		\end{equation*}}
	In summary we require that 
		\begin{equation*}
		\max\{ |\mathcal Q(\Pi_Kg) - \mathcal Q(g)| , \sup_{\phi \in \mathcal{G}_K} |\mathcal{Q}_n(\phi) -\mathcal Q(\phi)|    \} \big/ g_0(n,k, \epsilon) 
		\lesssim \tau_K\max\{ \sqrt{K/n} , \tau^{-1}_K \mynorm{\phi -g}_{L^2(Z)}    \} = o(1),
		\end{equation*}
		which follows from the rate restriction in Assumption \ref{A:Rate} (vi).
\end{proof}

 \bibliography{BiB}

\begin{thebibliography}{48}
\providecommand{\natexlab}[1]{#1}
\providecommand{\url}[1]{\texttt{#1}}
\expandafter\ifx\csname urlstyle\endcsname\relax
  \providecommand{\doi}[1]{doi: #1}\else
  \providecommand{\doi}{doi: \begingroup \urlstyle{rm}\Url}\fi

\bibitem[Abrevaya and Hausman(1999)]{AbrevHausman99}
J.~Abrevaya and J.~A. Hausman.
\newblock Semiparametric estimation with mismeasured dependent variables: an
  application to duration models for unemployment spells.
\newblock \emph{Annales d'Economie et de Statistique}, pages 243--275, 1999.

\bibitem[Abrevaya and Hausman(2004)]{abrevaya2004}
J.~Abrevaya and J.~A. Hausman.
\newblock Response error in a transformation model with an application to
  earnings-equation estimation.
\newblock \emph{The Econometrics Journal}, 7\penalty0 (2):\penalty0 366--388,
  2004.

\bibitem[Abrevaya and Shin(2011)]{abrevaya2011}
J.~Abrevaya and Y.~Shin.
\newblock Rank estimation of partially linear index models.
\newblock \emph{The Econometrics Journal}, 14\penalty0 (3):\penalty0 409--437,
  2011.

\bibitem[Ben-Moshe et~al.(2017)Ben-Moshe, D'Haultf{\oe}uille, and
  Lewbel]{benmosh2017}
D.~Ben-Moshe, X.~D'Haultf{\oe}uille, and A.~Lewbel.
\newblock Identification of additive and polynomial models of mismeasured
  regressors without instruments.
\newblock \emph{Journal of Econometrics}, 200\penalty0 (2):\penalty0 207--222,
  2017.

\bibitem[Berkson(1950)]{Berkson}
J.~Berkson.
\newblock Are there two regressions?
\newblock \emph{Journal of the American Statistical Association}, 45\penalty0
  (250):\penalty0 164--180, 1950.

\bibitem[Breunig and Haan(2018)]{breunig_haan2018}
C.~Breunig and P.~Haan.
\newblock Nonparametric regression with selectively missing covariates.
\newblock \emph{arXiv preprint arXiv:1810.00411}, 2018.

\bibitem[Breunig et~al.(2018)Breunig, Mammen, and Simoni]{breunigEndSel18}
C.~Breunig, E.~Mammen, and A.~Simoni.
\newblock Nonparametric estimation in case of endogenous selection.
\newblock \emph{Journal of Econometrics}, 202\penalty0 (2):\penalty0 268 --
  285, 2018.

\bibitem[Breunig et~al.(2019)Breunig, Huck, Schmidt, and
  Weizs\"acker]{HuckWeiz19}
C.~Breunig, S.~Huck, T.~Schmidt, and G.~Weizs\"acker.
\newblock The standard portfolio choice problem in germany.
\newblock \emph{CRC TRR 190 Discussion Paper}, \penalty0 (171), 2019.

\bibitem[Cavanagh and Sherman(1998)]{CavSherman}
C.~Cavanagh and R.~P. Sherman.
\newblock Rank estimators for monotonic index models.
\newblock \emph{Journal of Econometrics}, 84\penalty0 (2):\penalty0 351--381,
  1998.

\bibitem[Chen(2007)]{Chen07}
X.~Chen.
\newblock Large sample sieve estimation of semi-nonparametric models.
\newblock \emph{Handbook of Econometrics}, 2007.

\bibitem[Chen and Pouzo(2012)]{ChenPouzo12}
X.~Chen and D.~Pouzo.
\newblock Estimation of nonparametric conditional moment models with possibly
  nonsmooth generalized residuals.
\newblock \emph{Econometrica}, 80\penalty0 (1):\penalty0 277--321, 2012.

\bibitem[Chen et~al.(2003)Chen, Linton, and Van~Keilegom]{ChenShen2003}
X.~Chen, O.~Linton, and I.~Van~Keilegom.
\newblock Estimation of semiparametric models when the criterion function is
  not smooth.
\newblock \emph{Econometrica}, 71\penalty0 (5):\penalty0 1591--1608, 2003.

\bibitem[Chen et~al.(2005)Chen, Hong, and Tamer]{CHT_ME2005}
X.~Chen, H.~Hong, and E.~Tamer.
\newblock {Measurement Error Models with Auxiliary Data}.
\newblock \emph{The Review of Economic Studies}, 72\penalty0 (2):\penalty0
  343--366, 04 2005.

\bibitem[Chen et~al.(2011)Chen, Hong, and Nekipelov]{Chen_ME_Rev}
X.~Chen, H.~Hong, and D.~Nekipelov.
\newblock Nonlinear models of measurement errors.
\newblock \emph{Journal of Economic Literature}, 49\penalty0 (4):\penalty0
  901--37, December 2011.

\bibitem[Chiappori et~al.(2015)Chiappori, Komunjer, and Kristensen]{CKK}
P.-A. Chiappori, I.~Komunjer, and D.~Kristensen.
\newblock Nonparametric identification and estimation of transformation models.
\newblock \emph{Journal of Econometrics}, 188\penalty0 (1):\penalty0 22 -- 39,
  2015.

\bibitem[Clemencon et~al.(2008)Clemencon, Lugosi, and Vayatis]{ClemLugosi08}
S.~Clemencon, G.~Lugosi, and N.~Vayatis.
\newblock Ranking and empirical minimization of u-statistics.
\newblock \emph{The Annals of Statistics}, 36\penalty0 (2):\penalty0 844--874,
  2008.

\bibitem[D'Haultfoeuille(2010)]{2010Hault}
X.~D'Haultfoeuille.
\newblock A new instrumental method for dealing with endogenous selection.
\newblock \emph{Journal of Econometrics}, 154\penalty0 (1):\penalty0 1--15,
  2010.

\bibitem[Drerup et~al.(2017)Drerup, Enke, and von Gaudecker]{GaudeckerME}
T.~Drerup, B.~Enke, and H.-M. von Gaudecker.
\newblock The precision of subjective data and the explanatory power of
  economic models.
\newblock \emph{Journal of Econometrics}, 200\penalty0 (2):\penalty0 378 --
  389, 2017.

\bibitem[Dunker et~al.(2014)Dunker, Florens, Hohage, Johannes, and
  Mammen]{Dunker2011}
F.~Dunker, J.-P. Florens, T.~Hohage, J.~Johannes, and E.~Mammen.
\newblock Iterative estimation of solutions to noisy nonlinear operator
  equations in nonparametric instrumental regression.
\newblock \emph{Journal of Econometrics}, 178:\penalty0 444--455, 2014.

\bibitem[Fan et~al.(2020)Fan, Han, Li, and Zhou]{FanRank20}
Y.~Fan, F.~Han, W.~Li, and X.-H. Zhou.
\newblock On rank estimators in increasing dimensions.
\newblock \emph{Journal of Econometrics}, 214:\penalty0 379--412, 2020.

\bibitem[Han(1987)]{han1987}
A.~K. Han.
\newblock Non-parametric analysis of a generalized regression model: the
  maximum rank correlation estimator.
\newblock \emph{Journal of Econometrics}, 35\penalty0 (2-3):\penalty0 303--316,
  1987.

\bibitem[Hausman et~al.(1991)Hausman, Newey, Ichimura, and
  Powell]{Hausmanetal91}
J.~A. Hausman, W.~K. Newey, H.~Ichimura, and J.~L. Powell.
\newblock Identification and estimation of polynomial errors-in-variables
  models.
\newblock \emph{Journal of Econometrics}, 50\penalty0 (3):\penalty0 273 -- 295,
  1991.

\bibitem[Hoderlein and Winter(2010)]{HoderleinWinter2010}
S.~Hoderlein and J.~Winter.
\newblock Structural measurement errors in nonseparable models.
\newblock \emph{Journal of Econometrics}, 157\penalty0 (2):\penalty0 432 --
  440, 2010.

\bibitem[Hoderlein et~al.(2015)Hoderlein, Siflinger, and Winter]{hoderlein2015}
S.~Hoderlein, B.~Siflinger, and J.~Winter.
\newblock Identification of structural models in the presence of measurement
  error due to rounding in survey responses.
\newblock 2015.

\bibitem[Hu and Schennach(2008)]{hu2008}
Y.~Hu and S.~M. Schennach.
\newblock Instrumental variable treatment of nonclassical measurement error
  models.
\newblock \emph{Econometrica}, 76\penalty0 (1):\penalty0 195--216, 2008.

\bibitem[Imbens and Newey(2009)]{ImbensNewey09}
G.~W. Imbens and W.~K. Newey.
\newblock Identification and estimation of triangular simultaneous equations
  models without additivity.
\newblock \emph{Econometrica}, 77\penalty0 (5):\penalty0 1481--1512, 2009.

\bibitem[Jacho-Chavez et~al.(2010)Jacho-Chavez, Lewbel, and
  Linton]{JachoLewbel}
D.~Jacho-Chavez, A.~Lewbel, and O.~Linton.
\newblock Identification and nonparametric estimation of a transformed
  additively separable model.
\newblock \emph{Journal of Econometrics}, 156\penalty0 (2):\penalty0 392 --
  407, 2010.

\bibitem[Jureckova et~al.(2016)Jureckova, Koul, Navratil, and
  Picek]{JureckovaBernoulli16}
J.~Jureckova, H.~L. Koul, R.~Navratil, and J.~Picek.
\newblock Behavior of r-estimators under measurement errors.
\newblock \emph{Bernoulli}, 22\penalty0 (2):\penalty0 1093--1112, 2016.

\bibitem[Khan(2001)]{khan2001}
S.~Khan.
\newblock Two-stage rank estimation of quantile index models.
\newblock \emph{Journal of Econometrics}, 100\penalty0 (2):\penalty0 319--355,
  2001.

\bibitem[Lewbel(2014)]{Lewbel_SpecReg}
A.~Lewbel.
\newblock An overview of the special regressor method.
\newblock \emph{The Oxford Handbook of Applied Nonparametric and Semiparametric
  Econometrics and Statistics}, 2014.

\bibitem[Matzkin(1991)]{matzkin1991}
R.~Matzkin.
\newblock \emph{Nonparametric and Semiparametric Methods in Econometrics and
  Statistics}, chapter A Nonparametric Maximum Rank Correlation Estimator.
\newblock Cambridge: Cambridge University Press, 1991.

\bibitem[Matzkin(1994)]{matzkin1994}
R.~L. Matzkin.
\newblock Restrictions of economic theory in nonparametric methods.
\newblock \emph{Handbook of econometrics}, 4:\penalty0 2523--2558, 1994.

\bibitem[Matzkin(2007)]{matzkin2007}
R.~L. Matzkin.
\newblock Nonparametric identification.
\newblock \emph{Handbook of Econometrics}, 6:\penalty0 5307--5368, 2007.

\bibitem[Nadai and Lewbel(2016)]{NadaiLewbel}
M.~D. Nadai and A.~Lewbel.
\newblock Nonparametric errors in variables models with measurement errors on
  both sides of the equation.
\newblock \emph{Journal of Econometrics}, 191\penalty0 (1):\penalty0 19 -- 32,
  2016.

\bibitem[Newey et~al.(1999)Newey, Powell, and Vella]{Neweyetal99}
W.~Newey, J.~L. Powell, and F.~Vella.
\newblock Nonparametric estimation of triangular simulataneous equations
  models.
\newblock \emph{Econometrica}, 67\penalty0 (3):\penalty0 565--603, 1999.

\bibitem[Nolan and Pollard(1987)]{NolanPollard}
D.~Nolan and D.~Pollard.
\newblock U-processes: Rates of convergence.
\newblock \emph{The Annals of Statistics}, 15\penalty0 (2):\penalty0 780--799,
  1987.

\bibitem[Pakes and Pollard(1989)]{pakespollard89}
A.~Pakes and D.~Pollard.
\newblock Simulation and the asymptotics of optimization estimators.
\newblock \emph{Econometrica: Journal of the Econometric Society}, pages
  1027--1057, 1989.

\bibitem[Pollard(1984)]{Pollard84}
D.~Pollard.
\newblock \emph{Convergence of Stochastic Processes}.
\newblock Springer Series in Statistics, 1984.

\bibitem[Schennach(2007)]{Schennach07}
S.~Schennach.
\newblock Instrumental variable estimation of nonlinear errors-in-variables
  models.
\newblock \emph{Econometrica}, 75\penalty0 (1):\penalty0 201--239, 2007.

\bibitem[Schennach(2013)]{Schennach_rev}
S.~M. Schennach.
\newblock Measurement error in nonlinear models - a review.
\newblock \emph{Advances in Economics and Econometrics, Theory and
  Applications: Tenth World Congress of the Econometric Society}, 2013.

\bibitem[Sherman(1993)]{sherman93}
R.~P. Sherman.
\newblock The limiting distribution of the maximum rank correlation estimator.
\newblock \emph{Econometrica}, pages 123--137, 1993.

\bibitem[Shin(2010)]{shin2010}
Y.~Shin.
\newblock Local rank estimation of transformation models with functional
  coefficients.
\newblock \emph{Econometric Theory}, 26\penalty0 (6):\penalty0 1807--1819,
  2010.

\bibitem[Tang et~al.(2003)Tang, Little, and Raghunathan]{tang2003}
G.~Tang, R.~J. Little, and T.~E. Raghunathan.
\newblock Analysis of multivariate missing data with nonignorable nonresponse.
\newblock \emph{Biometrika}, 90\penalty0 (4):\penalty0 747--764, 2003.

\bibitem[van~der Vaart and Wellner(2000)]{Vaart2000}
A.~van~der Vaart and J.~Wellner.
\newblock \emph{{Weak Convergence and Empirical Processes: With Applications to
  Statistics (Springer Series in Statistics)}}.
\newblock Springer, corrected edition, Nov. 2000.

\bibitem[van~der Vaart and Wellner(2009)]{vdVaart_VCdim}
A.~van~der Vaart and J.~Wellner.
\newblock A note on bounds for vc dimensions.
\newblock \emph{IMS Collections: High Dimensional Probability}, 5:\penalty0
  103--107, 2009.

\bibitem[van~der Vaart(1998)]{vanVaart98}
A.~W. van~der Vaart.
\newblock \emph{Asymptotic statistics.}
\newblock Cambridge University Press, 1998.

\bibitem[White and Chalak(2010)]{white2010}
H.~White and K.~Chalak.
\newblock Testing a conditional form of exogeneity.
\newblock \emph{Economics Letters}, 109\penalty0 (2):\penalty0 88--90, 2010.

\bibitem[Zhao and Shao(2015)]{zhao2015}
J.~Zhao and J.~Shao.
\newblock Semiparametric pseudo-likelihoods in generalized linear models with
  nonignorable missing data.
\newblock \emph{Journal of the American Statistical Association}, 110\penalty0
  (512):\penalty0 1577--1590, 2015.

\end{thebibliography}
  
\end{document}